\documentclass[10pt,journal,compsoc]{IEEEtran}
\usepackage{hyperref}
\usepackage[caption=false]{subfig}
\usepackage{cite}
\usepackage{amsmath,amsthm,amssymb}
\usepackage{xparse}
\usepackage{url}
\usepackage{wrapfig}

\usepackage{graphicx}
\graphicspath{{./pdf/}}

\usepackage{tikz}
\usetikzlibrary{arrows,intersections}
\newtheorem{remark}{Remark}
\newtheorem{example}{Example}
\newtheorem{definition}{Definition}
\newtheorem{proposition}{Proposition}
\newtheorem{theorem}{Theorem}

\newtheorem{lemma}{Lemma}

\renewcommand{\Gamma}{\varGamma}
\renewcommand{\epsilon}{\varepsilon}
\newcommand{\etal}{\textit{et al}. }
\newcommand{\ie}{\textit{i}.\textit{e}. }
\newcommand{\eg}{\textit{e}.\textit{g}.}
\newcommand{\cf}{\textit{c}.\textit{f}. }
\newcommand{\DK}{\textsc{DK }}
\newcommand{\E}{\mathbf{E}}
\renewcommand{\Pr}{\mathbf{P}}
\newcommand{\Ubar}{\bar{u}}
\newcommand{\Umax}{u^{max}}
\newcommand{\I}{\mathcal{I}}

\begin{document}

	\title{Response time stochastic analysis for fixed-priority stable real-time systems}

	\markboth{~IEEE Transactions on Computers}%
	{Shell \MakeLowercase{\textit{et al.}}: Bare Demo of IEEEtran.cls for Computer Society Journals}
	
	\IEEEtitleabstractindextext{%
		\begin{abstract}
			In this paper, we prove that a mean system utilization smaller than one is a necessary condition for the feasibility of real-time systems. Such systems are defined as \textit{stable}. Stable systems have two distinct states: a transient state, followed by a steady-state where the same distribution of response times is repeated infinitely for each task. We prove that the Liu and Layland theorem holds for stable probabilistic real-time systems with implicit deadlines, we provide an analytical approximation of response times for each of those two states and a bound of the instant when a real-time system becomes steady.
		\end{abstract}
		
	}
	
	\author{\IEEEauthorblockN{Kevin~Zagalo\textsuperscript{$\ast$}, Yasmina~Abdeddaïm\textsuperscript{$\ast \dagger$}, Avner~Bar-Hen\textsuperscript{$\ast \ddagger$}, Liliana~Cucu-Grosjean\textsuperscript{$\ast$}} \\
		\IEEEauthorblockA{\textit{\textsuperscript{$\ast$}Inria~Paris, \textsuperscript{$\dagger$}Université~Gustave~Eiffel, LIGM, CNRS, \textsuperscript{$\ddagger$}Cnam~Paris}}}%
	
	\maketitle
	
	\IEEEraisesectionheading{\section{Introduction} \label{section:introduction}}
	
	\IEEEPARstart{P}{robabilistic} methods applied to real-time systems are motivated by the fact that worst-case execution times and minimal inter-arrival times are often rare events. The analysis induced by these methods consists in quantifying the probability that time requirements may not be satisfied. Their purpose is to soften real-time guarantees leading to an overestimated use of processors, or as we say, provide a pessimistic analysis~\cite{diaz2004pessimism}. Associating probabilities to worst-case values and to other possible cases may reduce the pessimism of a real-time analysis. Often, execution times have associated probabilities, but inter-arrival times may also be considered as probabilistic. In fact, many real-time systems have arrival times subject to noise and other interference from either operating systems \textit{etc.} which can be modeled with probability distributions~\cite{kdbench2022}. We call \textit{probabilistic real-time systems} the real-time systems which have \textit{probabilistic} parameters such as execution times and/or inter-arrival times.
	
	Utilization-based schedulability conditions for single-core fixed-priority preemptive scheduling policies are well known~\cite{davis2016review}. The seminal work of Liu and Layland \cite{liu1973scheduling} introduces a sufficient condition for the feasibility of a real-time system using its \textit{maximal utilization}. Nevertheless, a real-time system not satisfying this sufficient condition remains {\it schedulable} with a given probability (see \eqref{eq:schedulable} in Section~\ref{sec:taskmodel}). Moreover, while  probabilistic methods have been focused towards fitting in this sufficient condition by providing less pessimistic analyses, their domain of feasibility needs to be defined as well. In this paper, we build necessary feasibility conditions for fixed-priority scheduling policies based on the mean utilization of the real-time system. We demonstrate that a mean utilization smaller than $1$ is mandatory for \textit{response times} to be finite, \cf Propositions \ref{prop:feasible} and \ref{prop:LL}. We call systems with a mean utilization smaller than $1$, \textit{stable} real-time systems.

	Response times depend on properties of real-time systems such as the scheduling policy, the preemptiveness, \textit{etc.}. Our motivation is to exploit those properties leading response times to the domain of a certain probability distribution. 
	In this work, we emphasize the inverse Gaussian distribution as the appropriate distribution for an approximation of response times in the context of fixed-priority scheduling policies, using asymptotic results of queueing theory, \cf Propositions~\ref{prop:RTdistribution}. We propose two different approximations, a worst-case approximation before the system is in its \textit{steady-state} (\cf Proposition~\ref{prop:worstcase}, see Definition~\ref{def:steady}) and another one when the system is \textit{steady}, \cf Propositions \ref{prop:steadydistribution} and \ref{prop:backlog}. We prove in Proposition~\ref{prop:idleissteady} that a real-time system is in its steady-state at the first idle time.
	
	The contributions of this paper are:

	\begin{itemize}
		\item a necessary condition for the feasibility of probabilistic real-time systems with implicit deadlines under preemptive fixed-priority scheduling policies,
		\item an approximation of the demand of probabilistic real-time systems,
		\item a first proof that the Liu and Layland theorem \cite{liu1973scheduling} stands only for stable probabilistic real-time systems,
		\item a first proof that a real-time system is steady after the first idle time and an analytical formulation of its distribution,
		\item an analytical formulation of response time distributions of tasks before and after the first idle time,
		\item numerical evidence and comparison against simulated results and other existing estimations. 
	\end{itemize} 
	
	\begin{table*}[t]
	    \centering
	    \caption{Notations of this paper}
	    \begin{tabular}{cl|cl|cl}
            $\Gamma$ & task set & $\Pr, \E$ & probability and expectation & $C_{k,l}$ & execution time of $\tau_{k,l}$ \\
            $\tau_k$ & $k$-th task  &$\tau_{k,l}$ & $l$-th job of $\tau_k$ &$A_{k,l}$ & arrival time of $\tau_{k,l}$  \\
             $C_k$ & execution time of $\tau_k$ &$dF_k$ & execution time distribution of $\tau_k$ & $T_{k,l}$ & inter-arrival time between $\tau_{k,l-1}$ and $\tau_{k,l}$ \\
             $T_k$ & inter-arrival time of $\tau_k$ &  $m_k$ & mean execution time of $\tau_k$ & $R_{k,l}$ & response time of $\tau_{k,l}$ \\
             $D_k$ & relative deadline of $\tau_k$ & $dG_k$ & inter-arrival times distribution & $\hat{R}_{k,l}$ & heavy-traffic response time of $\tau_{k,l}$\\
             $N_k$ & counting process of jobs of $\tau_k$ & $\lambda_k$ & rate of arrivals of $\tau_k$ & $\tilde{R}_k$ & steady-state response time of $\tau_k$ \\
             $W_k$ & demand process of level $k$ & $\Ubar_k$ & mean utilization of level $k$ &  $\mathcal{I}^x$ & first idle time with backlog $x$ \\
             $\hat{W}_k$ & heavy-traffic demand of level $k$ &$u_k$ & mean utilization of $\tau_k$ & $\mathcal{I}_k^x$ & first idle time of level $k$ with backlog $x$  \\
             $\beta_k$ & backlog process of level $k$& $\Umax_k$ & maximum utilization of level $k$ & $d\mu_k$ & synchronous activation distribution \\
              $\hat{\beta}_k$ & heavy-traffic backlog of level $k$  & $v_k$ & demand deviation of level $k$ & $\tilde{h}_k$ & probability function of $\tilde{R}_k$ \\
             $\tilde{\beta}_k$ & steady-state backlog of level $k$& $\eta_k^{-1}$ & mean of $\tilde{\beta}_k - \tilde{\beta}_{k-1}$  & $h^{max}_k$ & probability function of $\sup_{l \in \mathbb{N}} \hat{R}_{k,l}$\\
             $d\pi_k$ & steady-state backlog distribution & $d\Phi$ & standard normal distribution & $\psi_k^x$ & probability function of $\I^x_k$\\
             $B_k(t)$ & blocking time of level $k$ at $t>0$ & $t_{idle}^{max}(\epsilon)$ & maximum $\epsilon$-idle time & $\psi$ & inverse Gaussian probability function  \\
        \end{tabular}
	    \label{tab:notations}
	\end{table*}
	
    The paper is organized as follows. We present in Section~\ref{sec:related} relevant existing results and in Section~\ref{sec:taskmodel} the main definitions of the considered model. We introduce in Section~\ref{sec:demand} the \textit{Loynes theorem} and the \textit{heavy-traffic theorem} which are the central theorems used this work. These theorems provide analytical expressions of the \textit{steadiness} of the system, through the stochastic analysis of \textit{backlogs}. In Section~\ref{sec:tda} we provide the time demand analysis associated to the probabilistic real-time systems. Moreover, we demonstrate the importance of the first idle time within any feasibility analysis. We detail the experimental results associated to the proposed contribution in Section~\ref{sec:exp} and we discuss future work opened by our current contribution in Section~\ref{sec:future}. Notations are provided in Table~1.

	\section{Related work}\label{sec:related}

    Probabilistic methods for the analysis of response times have many applications in real-time systems \cite{davis2019survey}. Two main directions have been explored: static methods for the exact computation and approximation of response time distributions  \cite{diaz2002stochastic, kim2005exact, maxim2013response, manolache2001memory} \textit{a priori}, and the measurement-based application of the extreme value theory (EVT) \cite{liu2013evt, lu2012statistical} approximating the distribution of the maximum values of response times \textit{a posteriori}. Often, probabilities are considered for execution times and few papers consider probabilistic inter-arrival times and deadlines~\cite{maxim2013response, lehoczky1996real, gaujal2020dynamic}. Moreover, the method introduced in \cite{diaz2002stochastic,kim2005exact} requires a large amount of convolutions which have a high space and time complexity, and the analysis provided by Lehoczky \cite{lehoczky1996real} is not suited to express response time distributions.

	The contributions of this paper are based on queueing theory results \cite{sparaggis1994optimal, sethuraman1999optimal, baccelli2013elements,chen2001fundamentals}. To the best of our knowledge, no result from the queueing theory is focused on general execution times and inter-arrival times, multi-class clients (\ie different tasks) and the quantization of deadline misses of such systems. The only results based on queueing theory for real-time systems have been published within the thread of papers related to \cite{lehoczky1996real}, where the author approximates  the number of simultaneously activated jobs by a \textit{reflected at the origin Brownian motion}. 
	However, the author makes strong hypotheses restricting the model. The exponential distribution of inter-arrivals and execution times suggested in \cite{lehoczky1996real} is a strong hypothesis. Furthermore, his model has another important limitation as it considers systems of jobs of only one task, which does not allow a response time analysis.
	
	\section{Model} \label{sec:taskmodel}
	
	We consider a probabilistic real-time system $\Gamma$ composed of a finite number of tasks ordered in the decreasing priority order, \ie a task $\tau_i$ has a higher priority than $\tau_j$ if and only if $i < j$. The tasks are scheduled on a single-core processor, using a  fixed-priority preemptive policy.

	\subsection{Probabilistic real-time tasks} \label{sec:probatask}
	A probabilistic task $\tau_k$ is characterized by: 

	\begin{enumerate}
		\item its \textit{execution time} $C_k > 0$ , with a distribution function $F_k$, and finite mean $m_k \triangleq \E[C_k]$ and non-negative variance $s_k^2 \triangleq \E[(C_k - m_k)^2]$,
		\item its \textit{inter-arrival time} $T_k > 0$  between two consecutive instances of $\tau_k$, with a distribution function $G_k$ of rate $\lambda_k  \triangleq1/\E[T_k]$,
		\item its \textit{relative implicit deadline} $D_k > 0$, with the same distribution function $G_k$ as its inter-arrival time.
	\end{enumerate}

	Without loss of generality, we consider the functions $F_k$ and $G_k$ discrete. Also note that all the parameters written in capital letters are probabilistic.
	
	\par The $l$-th instance of the task $\tau_k$ is called a \textit{job} and we denote it  $\tau_{k,l}$. Its execution time is $C_{k,l}$, with distribution function $F_k$. 	$T_{k,l}$ is the \textit{inter-arrival time} between the arrival times of the jobs  $\tau_{k,l-1}$ and $\tau_{k,l}$, with distribution function $G_k$, \ie identically distributed as $T_k$. We suppose all the rates of arrival $\lambda_k$ distinct.

    \par Due to the fixed-priority policy, the response times of the task $\tau_k$ depend only those of higher priority tasks. We call job of \textit{level} $k$ any job of a task of higher or equal priority than $\tau_k$, \ie any job $\tau_{i,j}, 1 \leq i \leq k$, $j \in \mathbb{N}$. Let $u_i \triangleq \E[C_i]/\E[T_i] = \lambda_i m_i$ be the \textit{mean utilization} of the task $\tau_i$,  $\Ubar_k \triangleq \sum_{i=1}^k u_i$ be the \textit{$k$-level mean utilization} of $\Gamma$, $\Ubar \triangleq \sum_{\tau_k \in \Gamma} u_k$ the \textit{total mean utilization} of $\Gamma$, and, the \textit{maximum utilization} of level $k$, $\Umax_k  \triangleq \sum_{i=1}^k c_i^{max} / t_i^{min} $, where $c_i^{max}$ is the maximum possible value of $C_i$ and $t_i^{min}$ is the minimum possible value of $T_i$, and the \textit{total maximum utilization}  $\Umax  \triangleq \sum_{\tau_k \in \Gamma} c_k^{max} / t_k^{min}$.
	
	\begin{definition} \label{def:stable}
		Let $\Gamma$ be a real-time system with total mean utilization $\Ubar$. $\Gamma$ is said \textit{stable} if and only if $\Ubar < 1$.
	\end{definition}

	The \textit{response time} $R_{k,l}$ of a job $\tau_{k,l}$ is the elapsed time between its activation (or arrival) $A_{k,l} \triangleq \sum_{j=1}^{l} T_{k,j}$ and the end of its execution. We consider implicit deadlines, \ie the absolute deadline of the job $\tau_{k,l}$ is $A_{k,l+1}$ and $A_{k,l+1} - A_{k,l} = T_{k,l+1} \sim D_k$ is its relative deadline with distribution function $G_k$. For the sake of comprehension, we refer to the relative deadline of $\tau_k$ with the variable $D_k$ and to its inter-arrival time $T_k$, even though they have the same distribution. At each arrival of jobs of a same task, previous jobs miss their deadline and are discarded.

	\begin{remark} \label{remark:discard} Such discarding policy is the reason why response times of jobs of a same task are independent: response times are functions of their associated execution times and only the higher priority jobs execution times. As described in \eqref{eq:TDA0}, discarding jobs makes their backlog disappear from the computation of the response time of the next job of the same task. However, in the response time analysis in Section~\ref{sec:tda} we build response time approximations regardless of this discarding policy which makes them upper-bounds of response times, see \eqref{eq:bound}. \end{remark}
    
	\begin{example}[Utilization and probabilistic deadlines] \label{ex:deadlines}
	    Let us consider the task set $\{ \tau_1, \tau_2 \}$, with $\Pr(C_1 = 1) = 1$, $\Pr(T_1=2)=1$, $\Pr(C_2 = 1) = 1/2$, $\Pr(C_2=2)=1/2$, $\Pr(T_2=3.1)=1/2$, $\Pr(T_2=4)=1/2$. Then $\Ubar = 1/2 + ((1+2)/2) / ((3.1+4/2) \approx 0.92$, $\Umax = 1/2 + 2/3.1 \approx 1.16$. Furthermore, because deadlines are implicit, if $T_{2,2} = 3.1$, the deadline of the job $\tau_{2,1}$ is $3.1$, if $T_{2,2} = 4$ its deadline is $4$. The deadline of any job $\tau_{1,j}$ is $2$ because $T_{1,j}$ is always $2$ for all $j \in \mathbb{N}$.
	\end{example}
	
    A job $\tau_{k,l}$ respects its deadline if its response time $R_{k,l}$ is smaller or equal to its deadline, \ie $R_{k,l} \leq D_k$. In order to be \textit{schedulable}, the \textit{worst-case deadline failure probability} \cite{von2021efficiently,navet2000worst} of a task $\tau_k$ should be smaller than its \textit{permitted failure rate} $\alpha_k \in (0,1)$ \cite{guo2015edf}, \ie \begin{equation}\small{  \Pr\left(\sup_{l \in \mathbb{N}} R_{k,l} > D_k\right)  = \int \Pr\left(\sup_{l \in \mathbb{N}} R_{k,l} > t \right) dG_k(t) \leq \alpha_k } \label{eq:schedulable} \end{equation} where $\alpha_k$ is a parameter given \textit{a priori}.
	
	\begin{definition} \label{def:feasible}
		Let $\Gamma$ be a real-time system. $\Gamma$ is said \textit{feasible} in the fixed-priority domain if and only if there exists at least one   fixed-priority preemptive scheduling policy such that all tasks $\tau_k \in \Gamma$ satisfy \eqref{eq:schedulable}.
	\end{definition}

	\begin{remark} \label{sec:indep}For all tasks $\tau_k \in \Gamma$ and all jobs $\tau_{k,l}, l \in \mathbb{N}$ the execution times $C_{k, l}$ are independent. We discuss in Section~\ref{sec:future} how the model may be extended in the dependent case, i.e. when the execution times $C_{k,l}, \tau_k \in \Gamma, l \in \mathbb{N}$ have a dependence structure.
	
	The inter-arrival times $T_{k, l}, l \in \mathbb{N}$ of a given task $\tau_k$ are independent, however inter-arrival times of two different tasks do not need to be independent for the provided results.
	
	For a given job $\tau_{k,l}$, its execution time $C_{k,l}$ and the inter-arrival time $T_{k,l}$ are not necessarily independent, meaning that the analysis we provide holds if execution times are correlated to the time separating a job to the previous instance of the same task. \end{remark}

    In the rest of this paper, we refer to \textit{probabilistic real-time systems} as simply \textit{real-time systems} when there is no ambiguity.
    
	\subsection{Elements of probability}
	Let us define some elements of probability that we use in this work. Let $X$ and $Y$ be two non-negative variables.
	
	\begin{definition}[Distribution function, probability function and distribution]
	     The \textit{distribution function} of $X$ is the function defined by $F_X(x) = \Pr(X \leq x)$. The \textit{probability function} of $X$ is the function $p_X(x) = (dF_X/dx)(x)$. We denote the \textit{distribution} of $X$ by $dF_X(x) = p_X(x)dx$.
	\end{definition}
	
	\begin{definition}[Identically distributed] \label{def:identical}
		$X$ and $Y$ are identically distributed and we write $X \sim Y$, or $X \sim dF_Y$, if and only if $F_X(x) = F_Y(x)$ for all $x > 0$.
	\end{definition}
	
	\begin{definition}[Independence] \label{def:independent}
		$X$ and $Y$ are independent if and only if $\Pr(X \leq x, Y \leq y)= F_X(x) \cdot F_Y(y)$ for all $x,y > 0$.
	\end{definition}

	\begin{definition}[Conditional distribution] \label{def:conditioned}
		The conditional distribution function of $X$ given $Y = y$ is written $$\Pr(X \leq x \mid Y = y) = \frac{1}{p_Y(y)} \cdot \frac{d}{dy}\Pr(X \leq x, Y \leq y)$$ for all $x,y > 0$, and is such that \begin{equation}\label{def:conditioneq}\Pr(X \leq x)= \int \Pr(X \leq x | Y = y) dF_Y(y)\end{equation} for all $x > 0$.
	\end{definition}

	\begin{definition}[Convolutions]
    Let $X$ and $Y$ be independent and $Z = X+Y$. Then the distribution function of $Z$ is the \textit{convolution} $F_Z = F_X \ast F_Y$ defined by $$F_Z(z) = \int F_X(z - y)dF_Y(y) = \int F_Y(z - x)dF_X(x)$$
    \end{definition}

	\begin{definition}[Brownian motions] \label{def:bm}
		A \textit{standard Brownian motion} is a process $B$ such that $B(0) = 0$ and with independent increments $B(t) - B(s)$ of distribution function $ \Phi\left( \cdot  /\sqrt{t-s}\right)$ for any $0 \leq s < t$, where $$d\Phi(x) =  \exp\left(-x^2/2\right) / \sqrt{2\pi} dx$$ is the standard normal distribution.
		
		A \textit{Brownian motion} of drift $u$ and deviation $v > 0$ is a process $W$ such that there exists a standard Brownian motion $B$ such that \begin{equation}W(t) = u t + v B(t) \label{eq:driftedbm}\end{equation} 
	\end{definition}
	
	\begin{definition}[Inverse Gaussian distribution]
	    Let $\theta > 0$ and $\gamma > 0$. The inverse Gaussian distribution of parameters $(\theta, \gamma)$ has a probability function $\psi$ defined by \begin{equation} \psi\left(t; \theta, \gamma\right) =  \sqrt{\frac{\gamma}{2\pi t^{3}}} \exp \left( - \frac{ \gamma( t - \theta)^2}{2 \theta^2 t} \right) \label{eq:inverseGaussienne} \end{equation}
	\end{definition}

    \begin{remark}
        Note that deterministic\footnote{A real-time system is deterministic if all task parameters are described by unique values.} real-time systems form a subclass of probabilistic real-time systems, as the distributions of the parameters of $\Gamma$ can describe a unique value if needed. Hence, we propose results covering also deterministic parameters, would they be execution times, inter-arrival times or deadlines. For instance, if the  execution time $C_k$ of a task $\tau_k \in \Gamma$ is deterministic and equal to $5$, then its distribution function is $F_k(x) = \mathbf{1}_{ [5, \infty )}(x)$, where \begin{equation*}\mathbf{1}_{A}(x) = \left\{\begin{array}{rl} 1 & \textrm{if } x \in A \\  0 & \textrm{otherwise} \end{array}\right.\end{equation*}
    \end{remark}

	\subsection{Backlog stochastic process} \label{sec:backlog}
	
	Let us define the following stochastic processes:
	
	\begin{enumerate}
		
		\item $N_{k}(t) \triangleq \sum_{l=1}^\infty \mathbf{1}_{ [0, t] }(A_{k, l}) $ as the number of jobs of $\tau_k$ released before $t > 0$, of mean $\E[N_k(t)] = \lambda_k t$ \cite[(1.1.22)]{baccelli2013elements}. $N_k$ is right-continuous.
		
		\item the $k$-level demand $W_k(t)$ as the accumulation of the execution times required by jobs of priority higher or equal than $\tau_k$, regardless of potential deadline misses, released before the instant $t$ to complete, \begin{equation*} W_k(t) \triangleq \sum_{i=1}^k \sum_{j=1}^{N_i(t)} C_{i,j} \end{equation*} of mean $\E[W_k(t)] = \sum_{i=1}^k \E[N_i(t)] \E[C_{i}] = \Ubar_k t$, see \cite[(6.53)]{Janssen2006}. $W_k$ is right-continuous.
		
		\item the $k$-level backlog $\beta_k(t)$ as the remaining demand at $t \geq 0$, defined by $\beta_k(0) \geq 0$ and \begin{equation}  \beta_k(t) \triangleq \beta_k(0) + W_k(t-) - \int_0^{t} \mathbf{1}_{ \{ \beta_k(s) > 0 \} } ds \label{def:backlogs}\end{equation} where $W_k(t-) = \lim_{\epsilon \to 0, \epsilon > 0} W_k(t - \epsilon)$ is the left limit of $W_k$ at the instant $t$, and $\int_0^{t} \mathbf{1}_{ \{ \beta_k(s) > 0 \} } ds$ is the total busy time of level $k$~\cite[p. 2]{lehoczky1990fixed} before $t > 0$.
	\end{enumerate}
	
	\begin{example}[Backlog]
	    Consider the task set of Example~\ref{ex:deadlines}. Suppose both tasks $\tau_1$ and $\tau_2$ are activated at time $t=0$ and  $T_{2,2}=3.1$, $C_{2,1}=2$, $C_{2,2}=1$. Then,  $\beta_2(0) = 0$,  $$\begin{array}{rcl} \beta_2(2) &=& \beta_2(0) + W_2(2-) - \int_0^2 \mathbf{1}_{ \{ \beta_2(s) > 0 \} } ds \\ &=& 0 + (1+2) - 2 = 1 \\ \beta_2(3) &=& \beta_2(0) + W_2(3-) - \int_0^3 \mathbf{1}_{ \{ \beta_2(s) > 0 \} } ds \\ &=& 0 + (1+2) - 3 = 0\\ \beta_2(3.1) &=& \beta_2(0) + W_2(3.1-) - \int_0^{3.1} \mathbf{1}_{ \{ \beta_2(s) > 0 \} } ds \\ &=& 0 + (1+2) - 3 = 0\\ \beta_2(3.2) &=& \beta_2(0) + W_2(3.2-) - \int_0^{3.2} \mathbf{1}_{ \{ \beta_2(s) > 0 \} } ds \\ &=& 0 + (1+2+1) - 3.1 = 0.9 \end{array}$$ 
	    At the instant $t=3.2$, the backlog of level $2$ is $0.9$.
	\end{example} 
	
	\begin{figure}
		\includegraphics[width=\linewidth]{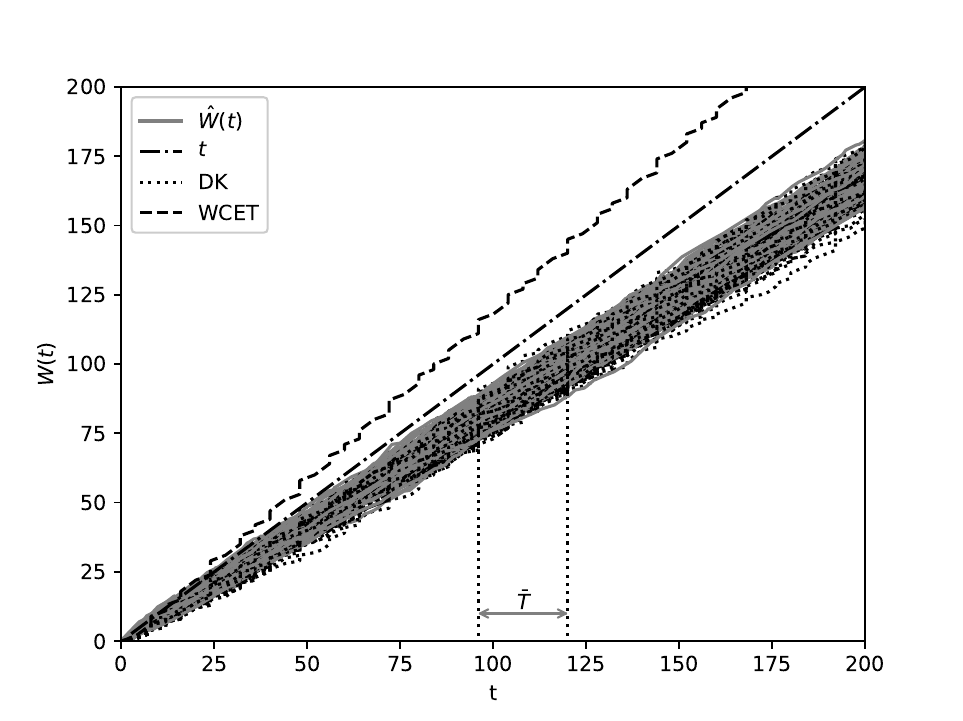}
		\caption{Level $3$ demand of $1\,000$ instances of the \DK model system, the heavy-traffic demand process and the classical deterministic worst-case analysis (WCET) considering only the maximal execution time for each task, for $\Ubar=0.838$ and $\Umax=1.208$ and hyper-period $\bar{T}$, for $\Gamma$ defined in Table~\ref{tab:taskset}.}
		\label{fig:diaz}
	\end{figure}
	
	The process $\beta_k$ describes the remaining demand without considering deadline misses, \ie while jobs are discarded their demand remains in the backlog analysis. Thus, the process $\beta_k$ is an upper-bound of the blocking time, see \eqref{eq:bound}, \ie the response time analysis that we provide in Section~\ref{sec:tda} is pessimistic as defined in \cite{diaz2004pessimism}.
	
	In stable real-time systems, there exist two main states of the backlog \cite[Section 4.2]{diaz2002stochastic}. 

	\begin{definition}[Transient and steady states] \label{def:steady}
		For a stable real-time system $\Gamma$, the \textit{steady-state} backlog is defined by $$\tilde{\beta}_k \triangleq \lim_{t \to \infty} \beta_k(t)$$ We say that $\Gamma$ is steady at the instant $t>0$ when its backlog is in steady-state, \ie $\beta_k(t) \sim \tilde{\beta}_k$, and transient before $\beta_k$ reaches its steady-state.
	\end{definition}
	
	In the following sections, we define the domain of this steadiness, provide an expression of the instant when the system goes from transient to steady, and prove that this time instant exists and is finite under some conditions.

	\section{Heavy-traffic theory} \label{sec:demand}
	
	In this section, we present an analytical approximation of the demand of probabilistic real-time systems. Let us denote $\pi_k(x) \triangleq \Pr\left(\tilde{\beta}_k \leq x\right)$ the distribution function of the steady-state backlog of level $k$. We use background results of queueing theory presented in \ref{sec:loynes} and \ref{sec:general}, provide the exact formulation of the steady-state backlog distribution $\pi_k$ in \ref{sec:exactdistribution} and illustrate this result in the deterministic inter-arrival case already studied by Diaz \etal \cite{diaz2002stochastic} in Section~\ref{blabla}.
	
	\subsection{The Loynes theorem} \label{sec:loynes}
	As the demand and backlog processes $W_k$ and $\beta_k$ are well studied in queueing theory \cite{baccelli2013elements, chen2001fundamentals, jeanblanc2009mathematical}, we provide the formula of the steady-state of the system, by adapting the \textit{Loynes theorem} \cite[Theorem 2.1.1]{baccelli2013elements}.
	
	\begin{theorem}[Loynes theorem] \label{thm:loynes} Let $A_l^k = \inf \{ t > 0 : \sum_{i=1}^k N_i(t) = l \}$ be the activation time of the $l$-th job of level $k$, and $I_l^k$  the index of the task of the $l$-th job of level $k$. If $\Ubar_k < 1$, then the steady-state backlog $\tilde{\beta}_k$ exists, is finite and we obtain \begin{equation} \tilde{\beta}_k =\sup_{n \geq 0}   \left( \sum_{l=1}^{n} C_{I_l^k, l} - (A_{l+1}^k - A_l^k)  \right)^+  \label{eq:steadyback}\end{equation} where $x^+ = \max(0, x)$, \cf \cite[Property 2.2.1]{baccelli2013elements}. In addition, there is an infinite number of idle times, \cf \cite[Property 2.2.5]{baccelli2013elements}. If $\Ubar_k = 1$, then the existence of a finite steady-state $\tilde{\beta}_k$ is uncertain. If $\Ubar_k > 1$, there exists a finite number of idle times of level $k$ and no finite steady-state, \cf  \cite[Property (2.2.2)]{baccelli2013elements}, backlogs are always transient.
	\end{theorem}

	Our goal is to find the distribution of the steady-state backlog $\tilde{\beta}_k$. However, given the generality of this model, we cannot provide an exact description of the process $\beta_k$. This is why we use the \textit{heavy-traffic assumptions} to find an exact formula for the distribution of the steady-state backlog $\tilde{\beta}_k$ when the system utilization gets close to $1$, \cf \figurename~\ref{fig:demand}, \ie we build a process $\hat{\beta}_k$ such that its steady-state is also $\tilde{\beta}_k$.
	
	\subsection{The heavy-traffic theorem} \label{sec:general}

	A first step in the approximation of the backlog process $\beta_k$ is the approximation of the demand process $W_k$. 
	
	\begin{theorem}[Heavy traffic theorem, Section 6.5, \cite{chen2001fundamentals} ]\label{lem:heavy} The level $k$ \textit{heavy-traffic} demand process \begin{equation*}\hat{W}_k(t)~\triangleq~\lim_{n \to \infty}~W_k(nt)/\sqrt{n} \end{equation*} is a Brownian motion of drift $\Ubar_k$ and deviation $v_k$, where $v_k^2~\triangleq~\sum_{i=1}^k~\lambda_i \cdot s_i^2$. See Figure 1.\end{theorem}

	\begin{figure*}[t]
		\centering
		\subfloat[$\Ubar = 0.625, \Umax=0.833$]{
			\includegraphics[width=0.23\textwidth]{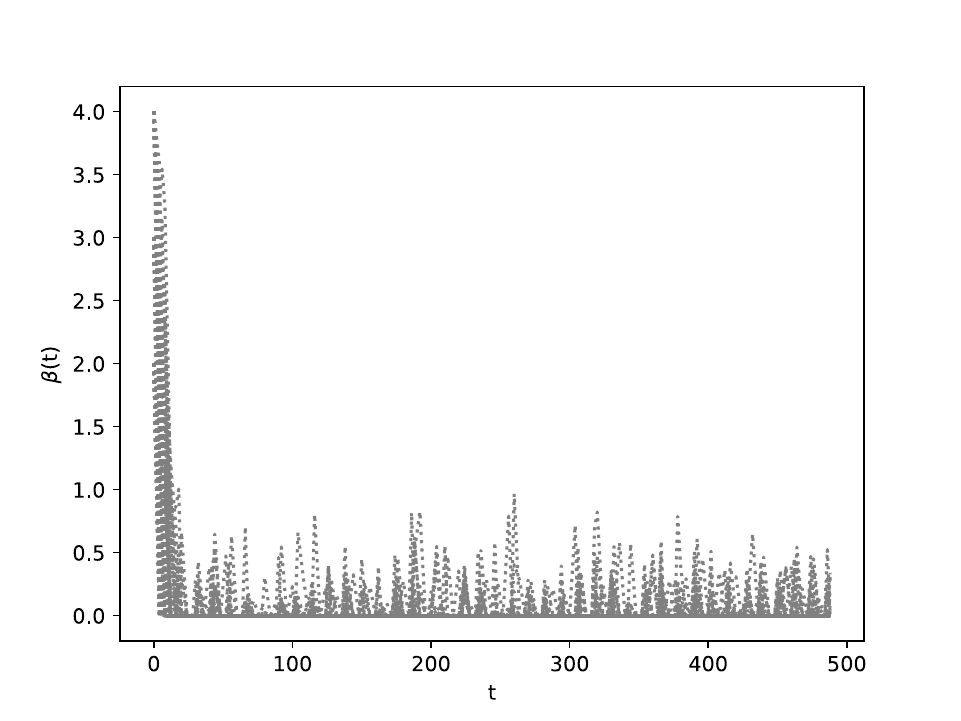}
			\label{fig:backlog1}
		}
		\hfill
		\subfloat[$\Ubar = 0.838, \Umax=1.208$]{
			\includegraphics[width=0.23\textwidth]{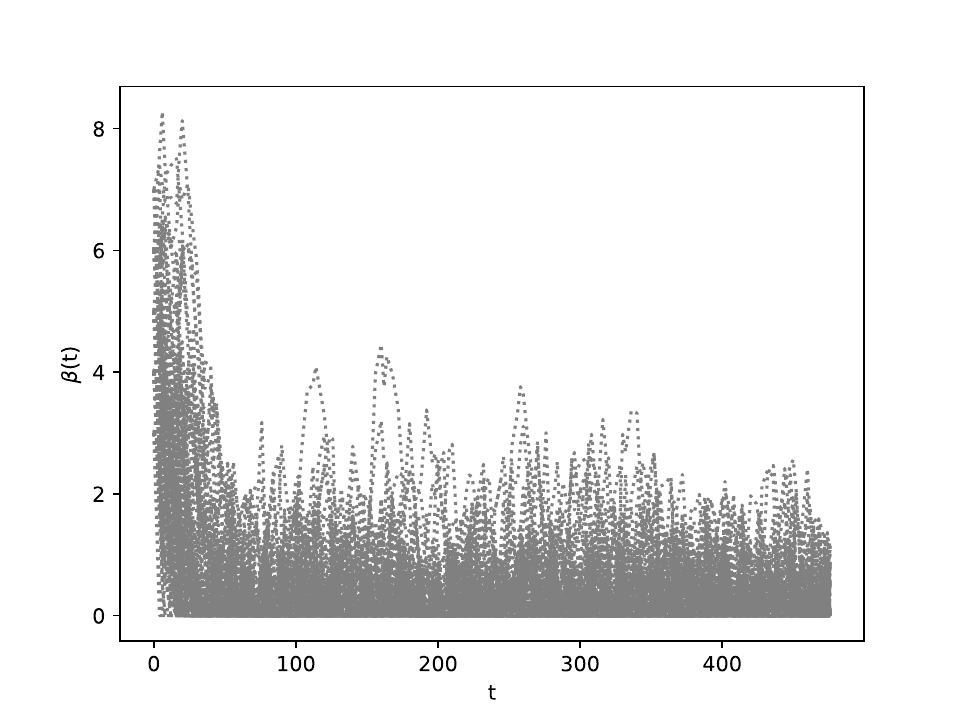}
			\label{fig:backlog2}
		}
		\hfill
		\subfloat[$\Ubar = 0.998, \Umax=1.508$]{
			\includegraphics[width=0.23\textwidth]{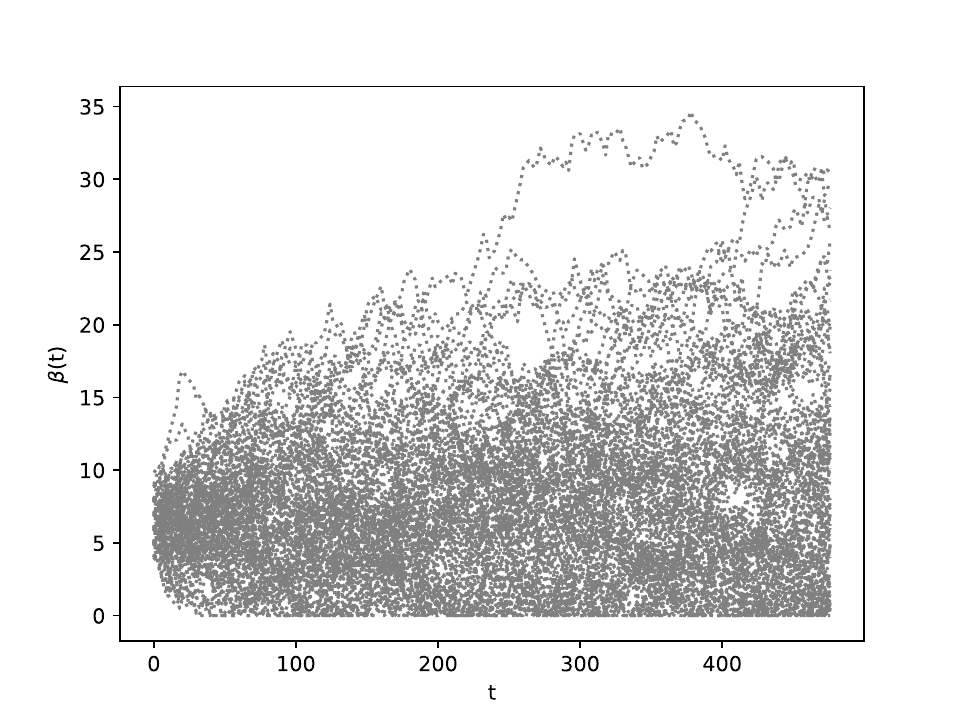}
			\label{fig:backlog3}
		}
		\hfill
		\subfloat[$\Ubar = 1.018, \Umax=1.608$]{
			\includegraphics[width=0.23\textwidth]{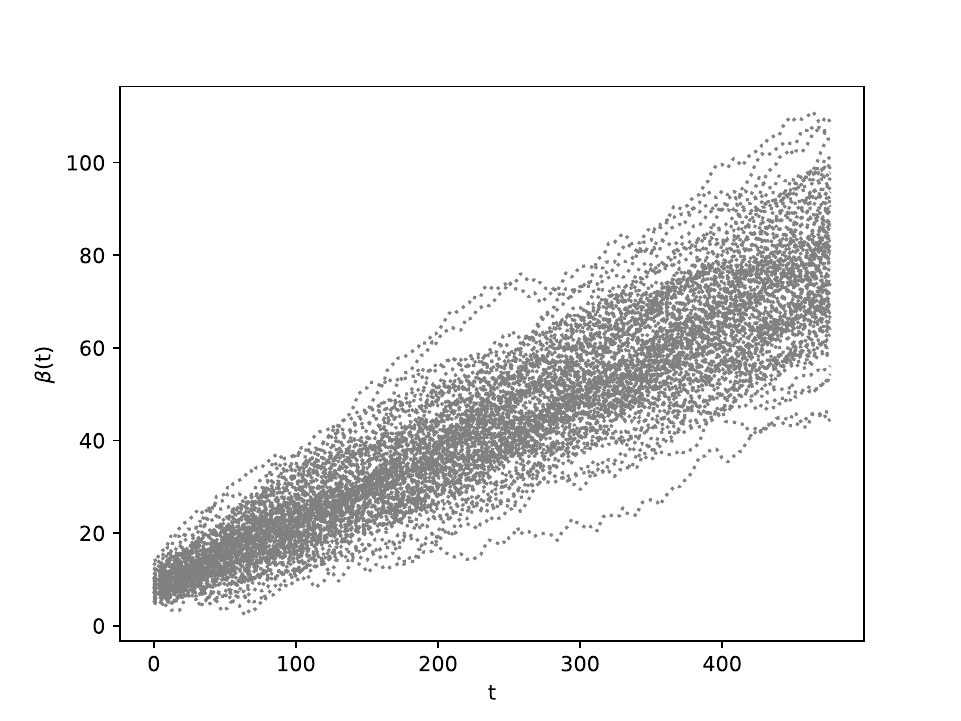}
			\label{fig:backlog4}
		}
		\caption{The heavy-traffic backlog process of systems with different mean utilization, initialized with $\beta(0) = \sum_{\tau_i \in \Gamma} C_i$.}
		
		\label{fig:demand}
	\end{figure*}

	In \cite{lehoczky1996real}, the author uses the \textit{heavy-traffic approximation} providing  the distribution of the lateness of jobs in a system with exponential inter-arrival times. To illustrate the heavy-traffic approximation, one can think of water continuously flowing into a sink at a rate $\lambda_i$ and the execution times as the rate $1/m_i$ the water leaves the sink. It is usually understood as true when the system is at full processor utilization, because the theorems of heavy-traffic theory are exact when $\Ubar \to 1$. However we use it as a way to build an upper-bound of the backlog process. Indeed, the heavy-traffic assumption should be seen as a bound, or more specifically, a way to suppose that the system utilization is at its  maximum (\ie $\Ubar = 1$), providing \emph{upper-bounds} that are exact when the processor utilization at 100\%. \figurename~\ref{fig:stationary2} illustrates this upper-bound becoming exact in \figurename~\ref{fig:stationary3} and \figurename~\ref{fig:stationary4}.

    In this section, we use the heavy-traffic assumption to find the steady-state backlog $\tilde{\beta}_k$ as an approximation upper-bounding the \textit{blocking time} of the system in its steady-state, see \eqref{eq:bound}. The approximation in Theorem~\ref{lem:heavy} leads to the standard Brownian motions which is continuous. It means that instead of looking at the demand for a large amount of time, we consider a re-scaled version of the demand in order to build a good approximation. Theorem~\ref{lem:heavy} can be written as  \begin{equation}\hat{W}_k(t) = \Ubar_k t + v_k B(t) \label{eq:drift} \end{equation} where $B$ is a standard Brownian motion, as defined in \eqref{eq:driftedbm}.
	
	\subsection{Steady-state backlog distribution} \label{sec:exactdistribution}

	In order to find the steady-state backlog $\tilde{\beta}_k$, we work with the heavy-traffic demand $\hat{W}_k$. We define the \textit{heavy-traffic backlog process} by \begin{equation}\hat{\beta}_k(t) =  \beta_k(0) + \hat{W}_k(t) - \int_0^{t} \mathbf{1}_{\{ \hat{\beta}_k(s) > 0 \}} ds \label{eq:driftedbacklog} \end{equation}
	
	This heavy-traffic backlog process $\hat{\beta}_k$ is a Brownian motion \textit{reflected at the origin} of drift $\Ubar_k - 1$ and deviation $v_k$, and its limit is also $\tilde{\beta}_k$, \cf \cite[Section 6.5]{chen2001fundamentals}. Its stationary distribution $\pi_k$ being the same for the initial backlog process $\beta_k$, \ie $\lim_{t \to \infty} \beta_k(t) = \tilde{\beta}_k = \lim_{t \to \infty} \hat{\beta}_k(t)$ we get its explicit expression in the following proposition.

	\begin{proposition} \label{prop:steadydistribution} Let $\epsilon > 0$ and $\Ubar = 1 - \epsilon$. Consider the parameter \begin{equation}\label{eq:etadeter}\eta_k \triangleq 2 \left(1 - u_k \right)/\left( \lambda_k (\gamma_{a,k}^2 + \gamma_{e,k}^2)\right)\end{equation} where $\gamma_{a,k}$ and $\gamma_{e,k}$ denote the coefficients of variation\footnote{The ratio of the standard deviation and the mean.} of respectively inter-arrival and execution times of $\tau_k$. Then, when $\epsilon \to 0$,  the distribution function $\pi_k$ of the steady-state backlog $\tilde{\beta}_k$ is defined for $x > 0$ by	\begin{equation}\label{eq:loynesdistribution}
        \pi_k(x) =   \sum_{i=1}^k \left(1 - \exp\left(- \eta_i x\right)\right) \prod_{\substack{j=1 \\ j\neq i}}^k  ( 1 - \eta_i/\eta_j)^{-1} 	\end{equation}
	\end{proposition}
	
	\begin{proof}
	   From  \cite[Section 6.5, p. 144]{chen2001fundamentals}, we know that when $\epsilon \to 0$, $\tilde{\beta}_k - \tilde{\beta}_{k-1}$ is an exponential variable of parameter $\eta_k$. As execution times are independent, $\tilde{\beta}_i - \tilde{\beta}_{i-1}$ and $\tilde{\beta}_j - \tilde{\beta}_{j-1}$ are independent for any $i \neq j$. Furthermore, $\tilde{\beta}_k =  \tilde{\beta}_1 + \sum_{i=2}^k \tilde{\beta}_i - \tilde{\beta}_{i-1}$, thus we know from \cite{ross2002probability} that the distribution of $\tilde{\beta}_k$ is $$d\pi_k(x) = \prod_{i=1}^k\eta_i \sum_{i=1}^k \frac{\exp\left(- \eta_i x\right)}{\prod_{j=1, j \neq i}^k (\eta_j - \eta_i)} dx$$

	for $x > 0$, which is the convolution of $k$ exponential distributions of parameters $\eta_1, \dots, \eta_k$. See \figurename~\ref{fig:backlog} for an illustration of the probability function of $\pi_k$.% This approximation is also called \textit{Kingsman's formula}.\\
	\end{proof}
	
	It is shown in \cite[Remark 6.17, p. 148]{chen2001fundamentals}, that the error of the estimation stated in \eqref{eq:driftedbacklog} is \begin{equation} \label{eq:error}
	\sup_{s \in [0, t]} \vert \beta_k(s) - \hat{\beta}_k(s) \vert = \mathcal{O}\left((t \log \log t)^{1/4} (\log t)^{1/2} \right)
	\end{equation} 

	Note that \eqref{eq:loynesdistribution} applies to the general case, hence the deterministic case, which means that the approximation of the stationary distribution in the periodic case is given in an explicit formula, without solving any system of equations neither computing a large number of convolutions.

	\begin{figure*}[t]
		\centering
		\subfloat[$\Ubar = 0.625, \Umax=0.833$]{
			\includegraphics[width=0.23\textwidth]{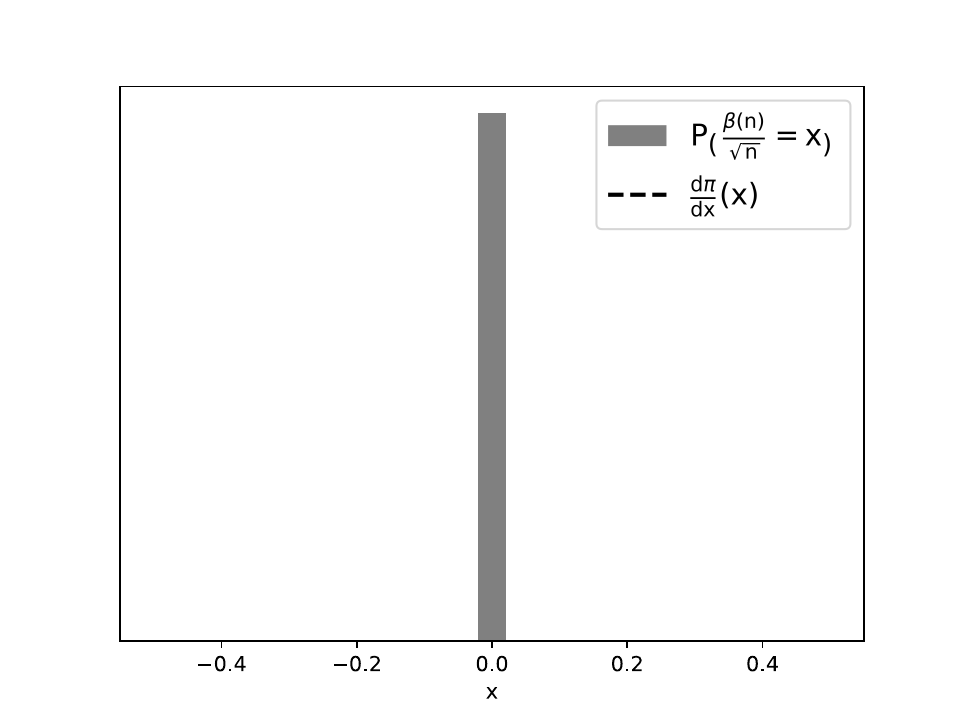}
			\label{fig:stationary1}
		}
		\hfill
		\subfloat[$\Ubar = 0.838, \Umax=1.208$]{
			\includegraphics[width=0.23\textwidth]{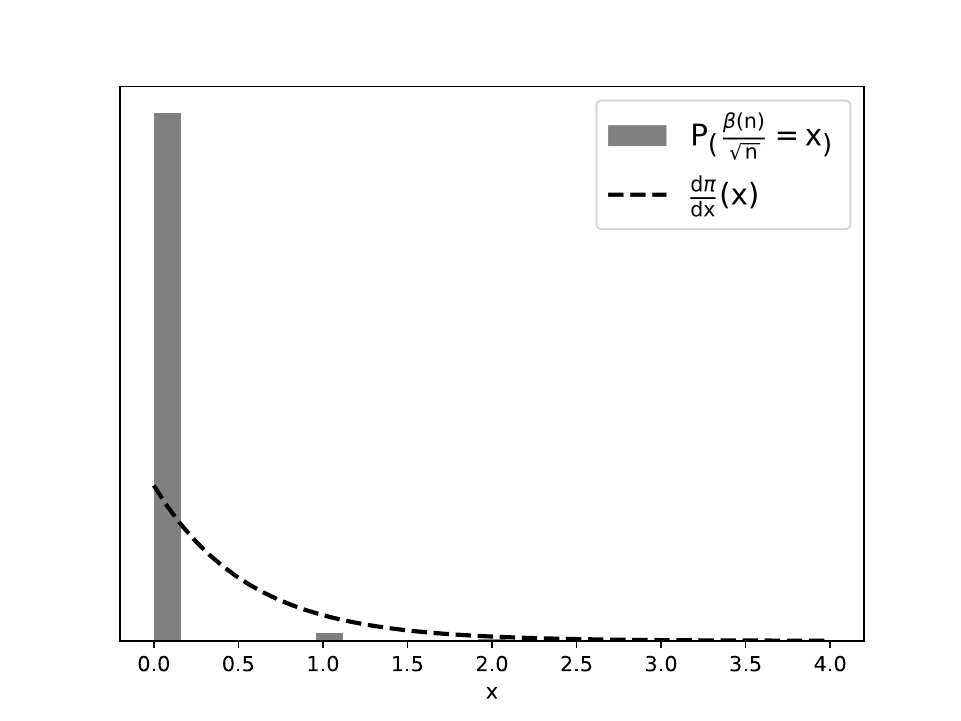}
			\label{fig:stationary2}
		}
		\hfill
		\subfloat[$\Ubar = 0.998, \Umax=1.508$]{
			\includegraphics[width=0.23\textwidth]{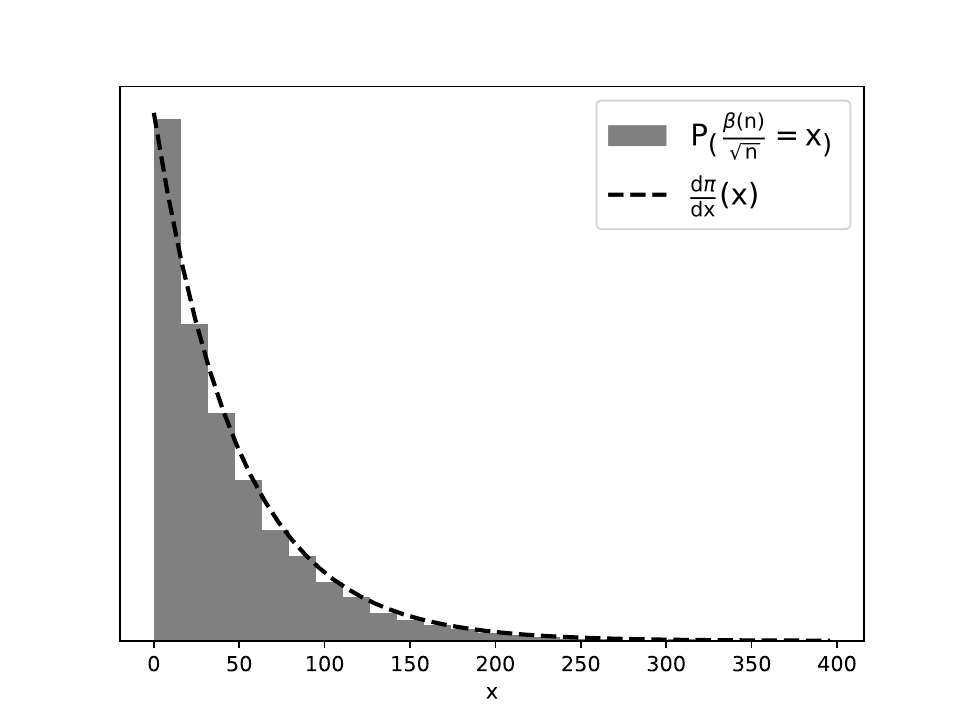}
			\label{fig:stationary3}
		}
		\hfill
		\subfloat[$\Ubar = 1.018, \Umax=1.608$]{
			\includegraphics[width=0.23\textwidth]{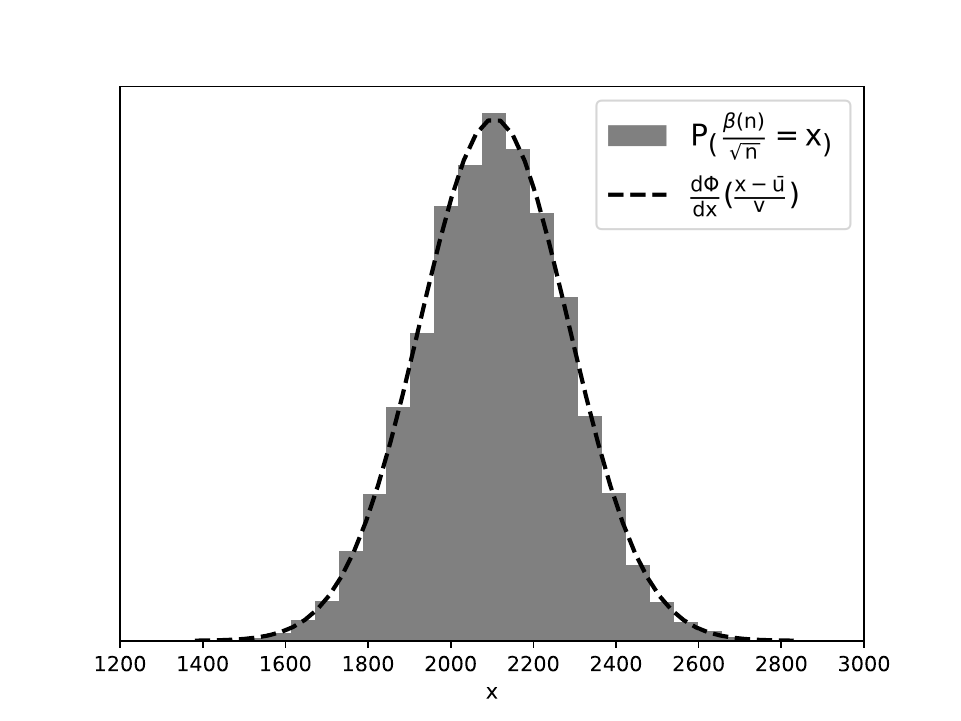}
			\label{fig:stationary4}
		}
		\caption{Steady-state backlog simulations with different mean utilizations in the DK model. In grey the histogram of simulations of $\beta(n)/\sqrt{n}$ for $n=10000$, in black the probability function of $\tilde{\beta}$.}
		\label{fig:backlog}
	\end{figure*}
	
		\subsection{The Diaz and Kim model: periodic workloads} \label{blabla}
	
	\par For the case where the tasks of the real-time system $\Gamma$ are periodic and have deterministic deadlines, \ie $T_k = \lambda_k^{-1} \in \mathbb{N}$ and  $G_k(x) = \mathbf{1}_{[\lambda_k^{-1}, \infty)}(x)$ for all $\tau_k \in \Gamma$ and $x > 0$, also known as the Diaz and Kim (DK) model~\cite{diaz2002stochastic, kim2005exact, diaz2004pessimism}, the authors approximate the distribution function $\pi_k$, resolving linear system equations and compute response times distributions with the help of convolutions.

	The fastest computational complexity of convolutions is $\mathcal{O}(N \log N)$ when $N$ is the number of values that a probabilistic variable can take. Computing the exact values of $\pi_k$ quickly becomes an expensive operation when the number of tasks or the number of possible execution times gets larger, even with methods providing tools to soften those computations like Markovic \etal \cite{markovic2021on},  Milutinovic \etal \cite{milutinovic2015speeding} or Maxim \etal \cite{maxim2012re} for example. Moreover, the computation of response times has the same problem, as the number of possible values of $\tilde{\beta}_k$ quickly becomes large. With \eqref{eq:steadyback} we have an explicit formula of the distribution of $\tilde{\beta}_k$, with \eqref{eq:error} we know the error of the heavy-traffic approximation, and with Theorem~\ref{thm:loynes} we have an analytical expression of the backlog in the deterministic case.
	
	Indeed, let $\bar{T}_k = \operatorname{lcm} \left(\lambda_1^{-1}, \dots, \lambda_k^{-1}\right)$ be the \textit{hyper-period} of level $k$ of $\Gamma$. In the DK model \cite{diaz2002stochastic, kim2005exact}, the authors consider the $k$-level backlog $\beta_k(t\bar{T}_k ), t \in \mathbb{N}$, \ie the remaining demand of level $k$ at the beginning of the $t$-th hyper-period. Diaz \etal \cite{diaz2002stochastic} have proven that the sequence $\left(\beta_k\left(t\bar{T}_k \right)\right)_{t \in \mathbb{N}}$ is a \textit{stationary} Markov chain when  $\Ubar_k < 1$. The sequence $\left(\beta_k(t\bar{T}_k )\right)_{t \in \mathbb{N}}$ is defined by $\beta_k(0) \geq 0$ and $\beta_k\left((t+1)\bar{T}_k\right) = \left( \beta_k(t\bar{T}_k ) + W_k(\bar{T}_k ) - \bar{T}_k  \right)^+$ for $t~\in~\mathbb{N}$.
	
	In order to apply \eqref{eq:steadyback} to the \DK model, we consider a system with only one task with an execution time equal to the total demand of one hyper-period $W_k(\bar{T}_k)$. Indeed, the demand of one hyper-period is repeated at each new hyper-period \cite{diaz2002stochastic}. Let $W_{k, t}, t \in \mathbb{N}$ be independent and identically distributed variables with the same distribution than $W_k(\bar{T}_k )$. Then from \eqref{eq:steadyback} we have \begin{equation}\tilde{\beta}_k \sim \underset{{n \geq 1}}{\sup} \left(\sum_{t=1}^n W_{k, t} - n \bar{T}_k  \right)^+ \label{eq:backlogdeterminist} \end{equation}

	The representation in \eqref{eq:backlogdeterminist} is an efficient method to approximate the stationary distribution $\pi_k$ of the backlog process $\beta_k$. Indeed, let us take an integer $n > 0$, and generate a sample $\left(W_{k,t}\right)_{t=1}^n$ independent and identically distributed sequence with the distribution of $W_k\left(\bar{T}_k\right)$. Eq. \eqref{eq:backlogdeterminist} provides the variable of distribution $\pi_k$ found in Diaz \etal \cite{diaz2002stochastic} and Kim \etal \cite{kim2005exact}. It also means that the variable $\max_{1 \leq j \leq n} \left(\sum_{t=1}^j \left(W_{k, t} - \bar{T}_k\right)  \right)^+$ is an approximation of $\tilde{\beta}_k$ when $n$ is large enough. This method is not expensive in complexity as it requires only to build the distribution function of $W_k\left(\bar{T}_k\right)$ once. 

	\begin{remark}
		In the \DK model, the coefficient of variation $\gamma_{a,k}$ is $0$, thus \begin{equation*}\eta_k = 2 m_k^2\frac{1-u_k}{\lambda_k s_k^2}\end{equation*}
	\end{remark}

	\section{Time demand analysis}  \label{sec:tda}
	
	In this section, we express the response times distribution of a task $\tau_k$ using the heavy-traffic analysis of Section~\ref{sec:demand}. The Markovian property of Brownian motions allows to approximate the distribution of any response time $R_{k,l}$ in terms of the backlog $\hat{\beta}_{k-1}(A_{k,l})$ and the first response time $R_{k,1}$, thus response times will be conditioned to backlogs and execution times, and represented as idle times following the inverse Gaussian distribution. In a second part, we provide an analytical expression of the worst-case heavy-traffic response time distribution, and in a third part we do the same for the steady-state heavy-traffic response time distribution. Finally, we explain how to simulate heavy-traffic response times of a task $\tau_k$ from the distribution functions $F_k$ and $G_k$.
	
    The response time $R_{k,l}$ of a job $\tau_{k,l}$ is the smallest instant after its arrival time $A_{k,l}$ smaller than the time it is given to run the level $k-1$ demand and its execution time, \ie the smallest $t >0$ such that  \begin{equation} B_{k-1}(A_{k,l}) + C_{k,l} +  W_{k-1}(t + A_{k,l}) - W_{k-1}(A_{k,l})  \leq t \label{eq:TDA0} \end{equation} where $B_{k}(t) = \sum_{i=1}^{k} \min\left(\beta_{i}(t) - \beta_{i-1}(t), C_{i}(t) \right)$ is the \textit{blocking time of level $k$} at the instant $t >0$, $\beta_0(t) = 0$, and $C_i(t)$ is the execution time of the most recent job of $\tau_i$ released before $t$, for all $t > 0$. This relation holds for any fixed-priority preemptive single-core model with the discarding policy discussed in Section~\ref{sec:indep}. By definition, the blocking time is such that \begin{equation}B_k(t) \leq \min\left(\beta_k(t), \sum_{i=1}^k C_i(t)\right) \label{eq:bound}\end{equation} for all $t>0$, which makes our response time analysis build upper-bounds of response times.
    
    Furthermore, the heavy-traffic demand $\hat{W}_{k-1}$ is a Brownian motion, thus  $\hat{W}_{k-1}(t+A_{k,l})-\hat{W}_{k-1}(A_{k,l}) \sim \hat{W}_{k-1}(t)$, and it is continuous. Then we define the \textit{heavy-traffic response time} of a job $\tau_{k,l}$ as \begin{equation}  \hat{R}_{k,l} \sim \inf_{t > 0} \left\{  \hat{\beta}_{k-1}(A_{k,l}) + C_{k,l} +  \hat{W}_{k-1}(t) = t \right\} \label{eq:TDA} \end{equation}
	
	In the following sections we consider response times as \textit{idle times} properly conditioned, and from this analysis provide the worst-case response time distribution and the steady-state response time distribution.
	
	\subsection{First idle time} \label{sec:relaxation}

	\begin{figure}[t]
		\centering
			\includegraphics[width=\linewidth]{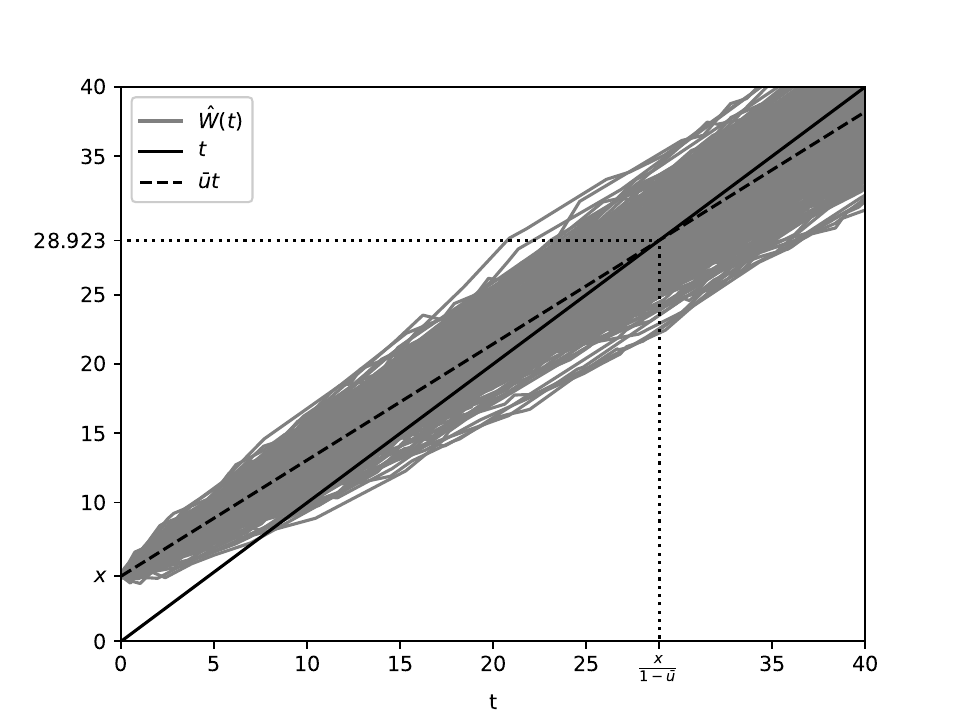}
			\label{fig:mean_response_time}
		\caption{Trajectories of $\hat{W}$ and average of the first idle time for $\Ubar = 0.838$.}
	\end{figure}

	In the following, we consider the conditional probability that the system starts with a level $k-1$ backlog $x \geq 0$ and the job $\tau_{k,1}$, $$\Pr_k^x( \cdot ) = \Pr( \cdot \mid \beta_{k-1}(0) = x, A_{k,1}=0)$$ Whenever we need to suppose $\beta_{k-1}(0) = x, A_{k,1}=0$, we say that we work \textit{under the probability} $\Pr^x_k$.

	As the steady state of backlogs repeats itself infinitely, the analysis of response time stands for all future response times. In the transient state though, each response time has its own distribution function. Hence, in Section~\ref{subsec:transient} we use the worst-case reasoning to bound the \textit{transient response times}. A question rises at this point: when does the system get steady ? Let us look closely at the dynamics of the heavy-traffic backlog process. Let $\hat{\beta}(t) \triangleq \max_{\tau_k \in \Gamma}  \hat{\beta}_k(t)$, $\hat{W}(t) \triangleq \max_{\tau_k \in \Gamma} \hat{W}_k(t)$ and $v \triangleq \max_{\tau_k \in \Gamma}  v_k$. As we work with fixed-priority scheduling policies, $\hat{\beta}, \hat{W}$ and $v$ refer to the lowest priority level.

	\begin{proposition} \label{prop:idleissteady}
	    Let $\Ubar < 1, x > 0$ and $\beta(0) = x$. The steady-state backlog is reached at the first idle time $$\I^x \triangleq \inf \left\{ t > 0 : \hat{\beta}(t)  = 0 \right\}$$
	\end{proposition}

    \begin{proof}
        For any $x>0$, at the instant $\I^x$, the next idle time has the distribution of $\I^0$, so is the third, the fourth, \textit{etc}. This is sufficient to say that the system is steady for all $t > \I^x$.  Moreover, one can see in \cite[Theorem 6.8]{chen2001fundamentals} that when $\Ubar < 1$ the steady-state backlog $\tilde{\beta}$ is $0$. 
    \end{proof}

	Proposition~\ref{prop:idleissteady} is coherent with the proven results in the deterministic case \cite{diaz2002stochastic} and the simulations shown in \figurename~\ref{fig:demand}. Let us now refine the analysis for each task $\tau_k \in \Gamma$. Let $\I_k^x$ be the \textit{first idle time of level $k$}, \ie \begin{equation} \I_k^x = \inf \left\{ t > 0 : \hat{\beta}_k(t)  = 0 \right\} \label{eq:defidle} \end{equation} when the initial backlog of level $k$ is equal to $x > 0$.
	
	\begin{figure}[t]
		\centering
			\includegraphics[width=\linewidth]{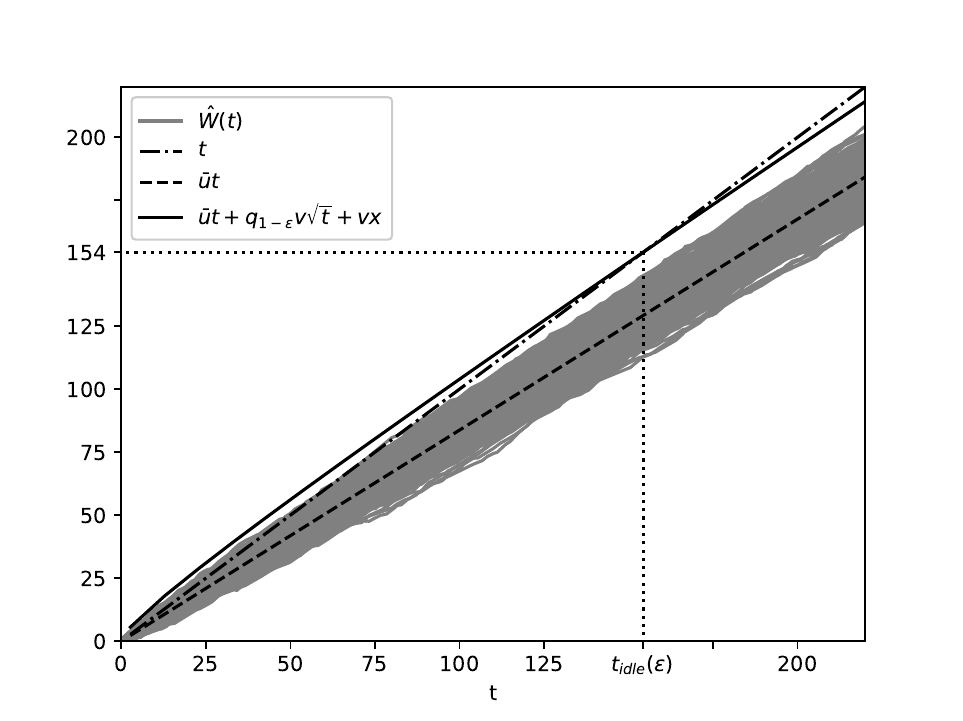}
			\label{fig:relaxation}
		\caption{First maximum $\epsilon$-idle time for $\epsilon=10^{-6}$. for $\Ubar = 0.838$.}
	\end{figure}
	
	\begin{lemma} \label{lem:inverseGauss}
	    The distribution of $\I_k^x$ is an inverse Gaussian distribution with probability function $$\psi_k^x(t) \triangleq \psi\left(t ; \frac{x}{1 - \Ubar_{k}}, \frac{x^2}{v^2_{k}}\right), t > 0$$
	\end{lemma}
	
	\begin{proof}
	    The idle time $\I_k^x$ defined in \eqref{eq:defidle} is a quantity called \textit{first-passage time} of a Brownian motion \cite[Eq. (28)]{molini2011first}. Until the first idle time, we know that $\hat{\beta}_k(t) = x + \hat{W}_k(t) - t$, because $\mathbf{1}_{\beta_k(s) > 0}=1$ for $0 \leq s < t \leq \I^x_k$. Then from \eqref{eq:drift}  we have $\I_k^x~\sim~\inf \left\{t > 0 : B(t) = t(1 -\Ubar_k) / v_k - x / v_k \right\}$ where $B$ is a standard Brownian motion. When $\beta_k(0) = x$, the distribution of $\I_k^x$ is an inverse Gaussian distribution of mean $x/(1-\Ubar_k)$ and shape $x^2/v_k^2$. See \figurename~\ref{fig:mean_response_time} and \cite[p. 146]{jeanblanc2009mathematical}, for more details.\end{proof}
	The first idle time is also a probabilistic variable, but we can build an instant $t_{idle}(\epsilon, x)$ such that for any $\epsilon > 0$ and all $t > t_{idle}(\epsilon, x)$, the probability that the first idle time of the lowest priority level is greater than $t$ is less than $\epsilon$, \ie $\Pr\left(\I^x > t \right) \leq \epsilon$. We call \textit{first $\epsilon$-idle time} the instant $t_{idle}(\epsilon, x)$ from which we can guarantee the system is steady with probability $1-\epsilon$ and initial backlog $x > 0$.  We know that $\hat{\beta}$ is a Brownian motion reflected at the origin. This means that until it reaches zero, \ie an idle time, it is a simple Brownian motion. In fact, between two consecutive idle times the backlog process has the same dynamic. The following proposition bounds with a given probability the first idle time and  the instant the system gets steady.
	
	\begin{proposition} \label{prop:trel} Let $\epsilon \in (0, 1)$ and $x \geq 0$. If $\Ubar < 1$, the system of initial demand $x$ is steady with probability at least $1 - \epsilon$ at the instant \begin{equation} t_{idle}(\epsilon, x) = \left( \frac{q_{1-\epsilon} + \sqrt{q_{1-\epsilon}^2 + 4\left(\frac{1-\Ubar}{v } \right) x}}{2\left(\frac{1-\Ubar}{v } \right)} \right)^2 \label{eq:relaxation} \end{equation} where $q_{1-\epsilon}=\Phi^{-1}(1-\epsilon)$ is the $(1-\epsilon)$-quantile  of a standard normal distribution.
	\end{proposition}
	
	\begin{proof}
	    Let us consider the lowest priority first idle time $\I^x$ and remark that \begin{align}\Pr(\I^x > t) &=  \Pr \left(\inf_{s \in [0, t]} \hat{W}(s) - s > -x  \right) \label{eq:idlek} \\ &\leq  \Pr\left(\hat{W}(t) > t-x \right) \nonumber \\ &= \Pr\left(B(t) > t \frac{1 -\Ubar}{v} - \frac{x}{v} \right) \nonumber \\ & = 1 - \Phi \left( \frac{(1-\Ubar) t - x}{v \sqrt{t}} \right) \nonumber \end{align}  Hence, let us define $t_{idle}(\epsilon, x)$ such that  \begin{equation*} 1 - \Phi \left( \frac{(1-\Ubar) t_{idle}(\epsilon, x) - x}{v \sqrt{t_{idle}(\epsilon, x)}} \right) \leq \epsilon \end{equation*} which implies that for $t > t_{idle}(\epsilon, x)$ we should have \begin{equation}t > \Ubar t + \Phi^{-1}(1-\epsilon) v \sqrt{t} + v x \label{eq:relaxtime} \end{equation} see \figurename~\ref{fig:relaxation}. This is a second order polynomial equation, which admits no solution when $\Ubar > 1$, and, when $\Ubar \leq 1$ and $\epsilon \in (0,1)$, the smallest solution of \eqref{eq:relaxtime} is \eqref{eq:relaxation}. 
	\end{proof}

	Note that for all $x \geq 0$, $\lim_{\epsilon \to 0} t_{idle}(\epsilon, x) = \infty$, $\lim_{\epsilon \to 1} t_{idle}(\epsilon, x) = 0$, and 
	\begin{itemize}
		\item $\Ubar < 1 \implies t_{idle}(\epsilon, x) < \infty$,
		\item $\Ubar = 1 \implies t_{idle}(\epsilon, x) = \infty$,
		\item $\Ubar > 1 \implies t_{idle}(\epsilon, x)$ is not defined.
	\end{itemize}
	
	\begin{remark}
		We provide the analysis for the lowest priority level, but in fact each priority level has its own first $\epsilon$-idle time.
	\end{remark}
	
	Finally we consider the \textit{maximum first $\epsilon$-idle time} $t^{max}_{idle}(\epsilon)$ by considering a synchronous activation for the distribution $\mu^{max}$ of the initial demand, \ie $\beta(0)=\sum_{\tau_k \in \Gamma} C_k$. See Proposition~\ref{prop:LL} for a more detailed explanation. The maximum first $\epsilon$-idle time is then \begin{equation}
	t_{idle}^{max}(\epsilon) = \int t_{idle}(\epsilon, x) d\mu^{max}(x) \label{eq:trelmax}
	\end{equation}

	\subsection{Conditioning response times}
	
	From \eqref{eq:TDA}, any job $\tau_{k,l}$ has a heavy-traffic response time distribution that can be expressed from the response time $\hat{R}_{k,1}$,
	\begin{equation}\Pr \left(\hat{R}_{k,l} > t \mid \hat{\beta}_{k-1}(A_{k,l}) = x \right) = \Pr^x_{k} \left(\hat{R}_{k,1} > t \right) \label{eq:markov} \end{equation} Let us condition this probability for specific  values of execution times.

	\begin{lemma}\label{eq:tda2} Under $\Pr^x_k$ the heavy-traffic response time $\hat{R}_{k,1}$ given $C_{k,1}=y$ is identically distributed as  $\I^{x+y}_{k-1}$.
	\end{lemma}
	
	\begin{proof}
		As stated in \eqref{eq:markov}, the proper conditioning on backlogs  provides the distribution of the response time of $\tau_{k,1}$. Furthermore, when the backlog $\hat{\beta}_{k-1}(A_{k,l}) = x$ and the execution time $C_{k,l}=y$, the response time $\hat{R}_{k,l}$ is the time it takes for all level $k-1$ jobs to finish plus the time it takes for level $k-1$ to stay idle for $x+y$ instants. It means that we can artificially set the initial backlog to $x+y$ and look at $\I_{k-1}^{x+y}$, the first idle time of level $k-1$, as represented in \figurename~\ref{fig:repidle}. In other words, \begin{multline}
		\Pr^x_{k} \left(\hat{R}_{k,1} > t \mid C_{k,1} = y \right) \\ = \Pr \left(\inf_{s \in [0, t]} \hat{W}_{k-1}(s) - s > -x-y \right)   = \Pr \left(\I_{k-1}^{x+y} > t\right) \nonumber
		\end{multline} which is sufficient to conclude.
	\end{proof}
	
	\begin{example}[Response times as idle times]\label{ex:repidle}
	    Consider a task set $\{\tau_1, \tau_2\}$, and that $\tau_2$ is activated at $t=0$, \ie $A_{2,1}=0$, and $C_{2,1}=y=3$. Let  $\Pr(C_1=1)=1/2, \Pr(C_1=2)=1/2$, $\Pr(T_1=2)=1/2$, $\Pr(T_1=4)=1/2$, and suppose $\beta_1(0) = x+y=8$, $C_{1,1} = 1$, $C_{1,2}=2$, $C_{1,3}=2$, and $A_{1,1}=~2, A_{1,2}=~6, A_{1,3}=10$ and $A_{1,4}=14$. Then the response time $R_{2,1}$ is the first instant $t$ when $\beta_1(t) = 0$: \begin{align*}\beta_1(1) & = \beta_1(0) = x + y\\
	    \beta_1(3) &= \beta_1(0)  + C_{1,1} - 3 = x + y - 2\\
	    \beta_1(7) &= \beta_1(0) + C_{1,1} + C_{1,2} - 7 \\& = x+y-4\\
	    \beta_1(11) &= \beta_1(0) + C_{1,1} + C_{1,2} + C_{1,3} - 11 \\ &= x+y-6\\
	    \beta_1(13) &= \beta_1(0) + C_{1,1} + C_{1,2} + C_{1,3} - 13 \\ &= x+y-8 = 0
	    \end{align*}
	    Hence $R_{2,1} = 13$. This is illustrated in \figurename~\ref{fig:repidle}. To build all possible values of $R_{2,1}$, one must do the same for all combinations of possible values of $(A_{1,1}, A_{1,2}, A_{1,3}, A_{1,4})$ and $(C_{1,1}, C_{1,2}, C_{1,3}, C_{1,4})$. Lemma~\ref{eq:tda2} provides an analytical expression of this result for the heavy-traffic response time $\hat{R}_{2,1}$.
	\end{example}

	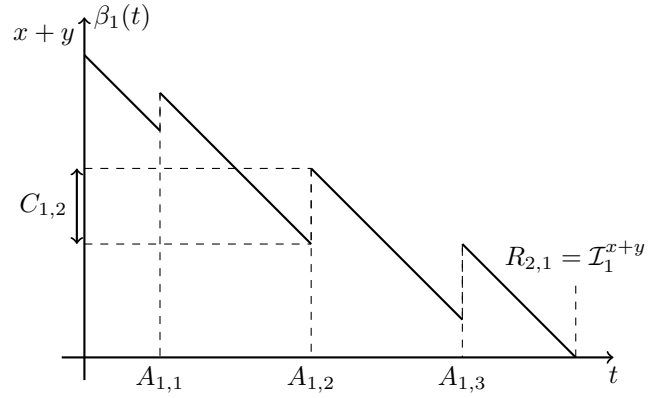
\begin{figure}[t]
		\centering
		\begin{tikzpicture}
            \coordinate (O) at (0,-1);
            \draw[->,thick] (-0.3,-1) -- (7,-1) coordinate[label = {below:$t$}] (xmax);
            \draw[->,thick] (0,-1.3) -- (0,3.5) coordinate[label = {right:$\beta_{1}(t)$}] (ymax);
            \path[name path=x] (0,3) -- (2,2);
            \draw[thick]      (1,2) -- (0,3) node[above left] {$x+y$};
            \draw[thick][thick]      (1,2.5) -- (3,0.5) node[pos=0.8, below right] {};
            \draw[dashed]       (1,2.5) -- (1, -1) node[below] {$A_{1,1}$};
            \draw[dashed]       (1,2) -- (1, 2.5) node[below] {};
            \draw[thick]      (3,1.5) -- (5,-0.5) node[pos=0.8, below right] {};
            \draw [dashed]      (3,1.5) -- (3,-1) node[below] {$A_{1,2}$};
        
            \draw [dashed]      (0,0.5) -- (3,0.5) node[above] {};
            \draw[<->,thick] (-0.1,0.5) -- (-0.1,1.5) node[pos=0.5, left] {$C_{1,2}$};
            \draw [dashed]      (0,1.5) -- (3,1.5) node[below] {};
        
            \draw[dashed]      (3,1.5) -- (3,0.5) node[below] {};
            \draw[thick]      (5,0.5) -- (6.5,-1) node[pos=0.8, below right] {};
            \draw [dashed]      (5,0.5) -- (5,-1) node[below] {$A_{1,3}$};
            \draw [dashed]    (5,-0.5) -- (5,0.5) node[below] {};
            \draw [dashed]      (6.5,-1) -- (6.5,0) node[above] {$R_{2,1} = \mathcal{I}_{1}^{x+y}$};
        \end{tikzpicture}
		\caption{Representation of the response time $R_{2,1}$ as an idle time when $A_{2,1}=0, C_{2,1}=y$ and $\beta_{1}(0)=x$, in Example~\ref{ex:repidle}.}
		\label{fig:repidle}
	\end{figure}
	
		From \eqref{eq:TDA} and Lemma~\ref{eq:tda2} we establish a necessary condition of the feasibility of $\Gamma$.
	
	\begin{proposition} \label{prop:feasible}
		A non-stable real-time system with implicit deadlines and independent execution times is not feasible under a fixed-priority scheduling policy.
	\end{proposition}
	
	\begin{proof}
		Let $\tau_k \in \Gamma$ be such that $\Ubar_{k} > 1$. The Loynes theorem \ref{thm:loynes} states that there is a finite number of idle times of level $k$, which implies with \eqref{eq:TDA} that heavy-traffic response times get infinite at some point, \ie for all $t> 0$, $\Pr\left(\sup_{l \in \mathbb{N}} \hat{R}_{k,l} > t \right) = 1$. In other words, there is  no permitted failure rate $\alpha_k \in (0,1)$ such that $\tau_k$ is schedulable as defined in \eqref{eq:schedulable}. As we consider a fixed-priority scheduling policy, then the lowest level backlog is larger than all $k$-level backlogs $\beta_k(t)$ at any time $t>0$.
	\end{proof}

	With the same reasoning we get the explicit formula of the distribution function of response times.
	
	Let us denote the conditional probability that the distribution of the initial backlog $\beta_{k-1}(0)$ is the probability $d\mu$ and the first job released is $\tau_{k,1}$, $$\Pr_k^\mu(\cdot) \triangleq \Pr(\cdot \mid \beta_{k-1}(0) \sim d\mu, A_{k,1}=0) =\int  \Pr^x_k(\cdot) d\mu(x)$$

	\begin{lemma} \label{coro:inverse}
		Under $\Pr^\mu_k$, the probability function of the heavy-traffic response time $\hat{R}_{k,1}$ is $$h_k^\mu(t) = \int \psi_{k-1}^{x}(t)d(\mu \ast F_k)(x)$$ 
	\end{lemma}

	\begin{proof}
		We conclude with  Lemmas \ref{lem:inverseGauss} and \ref{eq:tda2} and the definition of conditional probabilities, \cf \eqref{def:conditioneq}.
	\end{proof}

	In Sections~\ref{subsec:transient} and \ref{subsec:steady}, we prove that the proper initialization of the system puts the system in two specific cases: the worst-case and the steady-state. 
	
	\subsection{Worst-case response time} \label{subsec:transient}
	
	Before the system reaches its steady-state, we say it is \textit{transient}. In that case we cannot provide the exact distribution of the backlog in an analytical formulation. However, we can bound it by using the {\it worst-case backlog}.
	
	Let $\mu_k = F_1 \ast \dots \ast F_k$ be the distribution function of the demand of level $k$ of a \textit{synchronous activation} $\sum_{i=1}^{k} C_i$.
	\begin{proposition} \label{prop:LL}
		If $\Ubar_{k-1} < 1$, the worst-case heavy-traffic response time distribution of $\tau_k$ is defined by \begin{equation} \Pr\left(\sup_{l \in \mathbb{N}} \hat{R}_{k,l} > t \right) =  \Pr_{k}^{\mu_{k-1}}\left(\hat{R}_{k,1} > t\right) \label{eq:liu} \end{equation} 
	\end{proposition}
	
	\begin{proof} 
		First, let us consider $\Ubar_{k-1} < 1$, as stated in Theorem~\ref{thm:loynes} the backlog process converges to $\tilde{\beta}_k$ which is finite. Jobs are discarded if they miss their deadlines, and as we consider implicit deadlines, there can be at most one job per task activated simultaneously. This leads into considering $\beta_{k-1}(0) = \sum_{i=1}^{k-1} C_i$ as the maximum backlog of level $k-1$, \cf \eqref{eq:bound}, and use the property stated in \eqref{eq:markov}. 
		
		When $\Ubar_{k-1} = 1$, idle times of level $k$ may or may not be finite, thus we cannot conclude anything on the distribution of response times.
		
		When $\Ubar_{k-1} > 1$, the largest response time does not come from the synchronous activation and is in fact $\infty$. As we have already demonstrated in the proof of Proposition~\ref{prop:feasible}, response times of the task $\tau_k$ increase to $\infty$, due to the absence of idle times of level $k-1$. Then we conclude that the worst-case heavy-traffic response time is $\infty$.
		
		See \figurename~\ref{fig:backlog}. 
	\end{proof}
	
	The worst-case response time being identified as the synchronous activation case, we provide an explicit expression its distribution function. 
	
	\begin{proposition} \label{prop:worstcase}
		Let $\tau_k \in \Gamma$ and $\Ubar_k < 1$. The probability function of the worst-case heavy-traffic response time of $\tau_k$ is \begin{equation} h_k^{max}(t) \triangleq \int_0^\infty \psi_{k-1}^x(t) d\mu_{k}(x)  \label{eq:steadystate2} \end{equation}
	\end{proposition}
	
	\begin{proof}
		From Proposition~\ref{prop:LL} we know that when $\Ubar_k < 1$, the worst-case response time of $\tau_k \in \Gamma$ is $\hat{R}_{k,1}$ under $\Pr_k^{\mu_{k-1}}$. Furthermore, from Lemma~\ref{coro:inverse} we have $h_k^{max}(t)  = \int \psi_{k-1}^{z} \left(t \right) d(\mu_{k-1} \ast  F_k) (z) $. In addition, we have the relation $\mu_{k-1} \ast F_k = \mu_k$.
	\end{proof}
	
	\subsection{Steady-state response time} \label{subsec:steady}
	
	The backlog process of level $k$ being stationary and with a stationary distribution $\pi_k$, in the steady-state \eqref{eq:markov} becomes \begin{equation}\Pr_{k}^{\pi_{k-1}} \left(\hat{R}_{k,l} > t \right) = \Pr_{k}^{\pi_{k-1}} \left(\hat{R}_{k,1} > t \right) \label{eq:markovrt}\end{equation} for any $l \in \mathbb{N}$. Then if it holds for any $l \in \mathbb{N}$, it holds for all the response times after the convergence of the backlog process. This is why in the steady-state, the distribution of heavy-traffic response times is unique. Consider the steady-state response time \begin{align} \tilde{R}_k \triangleq \inf \left\{ t > 0 : \tilde{\beta}_{k-1} + C_k +  \hat{W}_{k-1}(t) = t \right\} \label{eq:TDAsteady} \end{align}
	
	As in Section~\ref{subsec:transient}, we get the distribution function of steady-state response times.

	\begin{proposition} \label{prop:backlog}
		Let $\tau_k \in \Gamma$, $\Ubar_k < 1$, $\tilde{h}_k$ be the probability function of the steady-state response time of $\tau_k$, and for $i \in \{1, \dots, k\}$ let $\eta_i$ be as defined in Proposition~\ref{prop:steadydistribution} and $\xi_{i,k} = \prod_{j=1, j\neq i}^{k}  ( 1 - \eta_i/\eta_j)^{-1}$. Then for all $t>0$,\begin{equation*} \tilde{h}_k(t) =  \sum_{i=1}^{k-1} \xi_{i,k-1} \int_0^\infty \psi^z_{k-1}(t)  d\nu_{i,k}(z) \label{eq:steadystate} \end{equation*} where $d\nu_{i,k}(z)=\eta_i \int_0^z  e^{- \eta_i (z-x)} dF_k(x) dz$ is the distribution  of $C_k + \tilde{\beta}_i - \tilde{\beta}_{i-1}$.
	\end{proposition}
	
	\begin{proof}
	    We know from Proposition~\ref{prop:steadydistribution} that the steady-state level $(k-1)$ backlog distribution is $\pi_{k-1}$. We know from Lemma~\ref{eq:tda2} that the stead-state response time $\tilde{R}_k$ is the first idle time with an initial backlog of $\tilde{\beta}_{k-1} + C_k$. Finally from Lemma~\ref{coro:inverse}, we know that the probability function of $\tilde{R}_k$ is $\tilde{h}_k \triangleq h^{\pi_{k-1}}_k$. Then, let us compute the convolution $\pi_{k-1} \ast F_k$ which is the distribution function of the sum $\tilde{\beta}_{k-1} + C_k$.  For any $z > 0$, we have $d(\pi_{k-1} \ast  F_k) (z) = \frac{d\pi_{k-1}}{dz} \ast  F_k (z)dz$ and
	    \begin{align*}
	         \frac{d\pi_{k-1}}{dz} \ast  F_k&(z)dz = dz\int_0^\infty \frac{d\pi_{k-1}}{dz}(z - x) dF_k(x) \\
	       =& dz \int_0^z  \eta_1 \dots \eta_{k-1} \sum_{i=1}^{k-1} \frac{\exp\left(- \eta_i (z-x)\right)}{\prod_{j=1, j \neq i}^{k-1} (\eta_j - \eta_i)} dF_k(x)\\
	                              =& dz\int_0^z  \sum_{i=1}^{k-1} \eta_i e^{- \eta_i (z-x)} \xi_{i,k-1} dF_k(x)\\
	                              =& \sum_{i=1}^{k-1} \xi_{i,k-1} d\nu_{i,k}(z)
	                              %& = \sum_{i=1}^{k-1} \eta^{\times i }_{k-1} e^{- \eta_i z} \E\left[ e^{\eta_i C_k} \mathbf{1}_{[0,z]}(C_k) \right]
	    \end{align*}
		Finally by definition of conditional probabilities, \cf \eqref{def:conditioneq}, we have $\tilde{h}_k(t) = \int \psi_{k-1}^{z} \left(t \right) d(\pi_{k-1} \ast  F_k) (z) $. %For \eqref{eq:meanRT}, the mean value of the seady-state response time,  we only need to remark that $\E[\tilde{\beta}_{k-1}] = 1 / \eta_{k-1} = \frac{v_{k-1}^2/2}{1-\Ubar_{k-1}}$. Hence, the average steady-state response time is the first idle time of level $k-1$ with an initial work load of $x = m_k + 1 / \eta_k$. Eq. \eqref{eq:bachelier} gives us the  distribution of this idle time of mean  $\frac{x}{1 - \Ubar_{k-1}}$, \cf \eqref{eq:idleavg}. Eq. \eqref{eq:meanRT} is a variant of the Phipp's formula adapted to real-time systems, see (3.5.17), page 242, in \cite{baccelli2013elements}.
	\end{proof}

	\begin{remark}
		In the periodic case,  Diaz \cite{diaz2002stochastic} has proven that the system is always steady when $\Umax_k \leq 1$, \eg \  \figurename~\ref{fig:stationary1}, \ie  the steady-state backlog is $\tilde{\beta}_k = 0$, and that is it reached at $t=0$. \begin{equation} \tilde{h}_k(t) = \int \psi^x_{k-1}(t)dF_k(x) \label{eq:steadystate0} \end{equation}
			This steadiness is illustrated in \figurename~\ref{fig:demand}.
	\end{remark}

	\section{Simulations} \label{sec:exp}

	\begin{figure*}[t]
		\centering
		\subfloat[Histogram of the generated response times of $\tau_3$ in Table~\ref{tab:taskset}, as defined in Section~\ref{sec:sim}.]{\label{fig:sim_response_time}
			\includegraphics[width=0.48\linewidth]{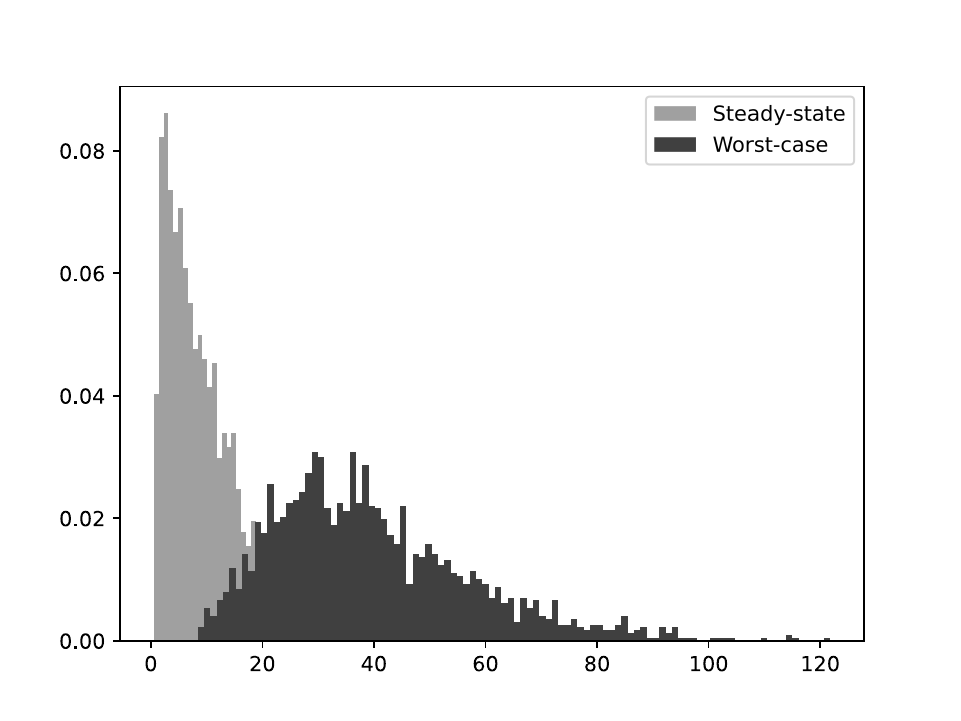}
		}
		\subfloat[Response time distribution functions of $\tau_3$ in Table~\ref{tab:taskset}.\label{fig:rt_cdf}]{
			\includegraphics[width=0.48\linewidth]{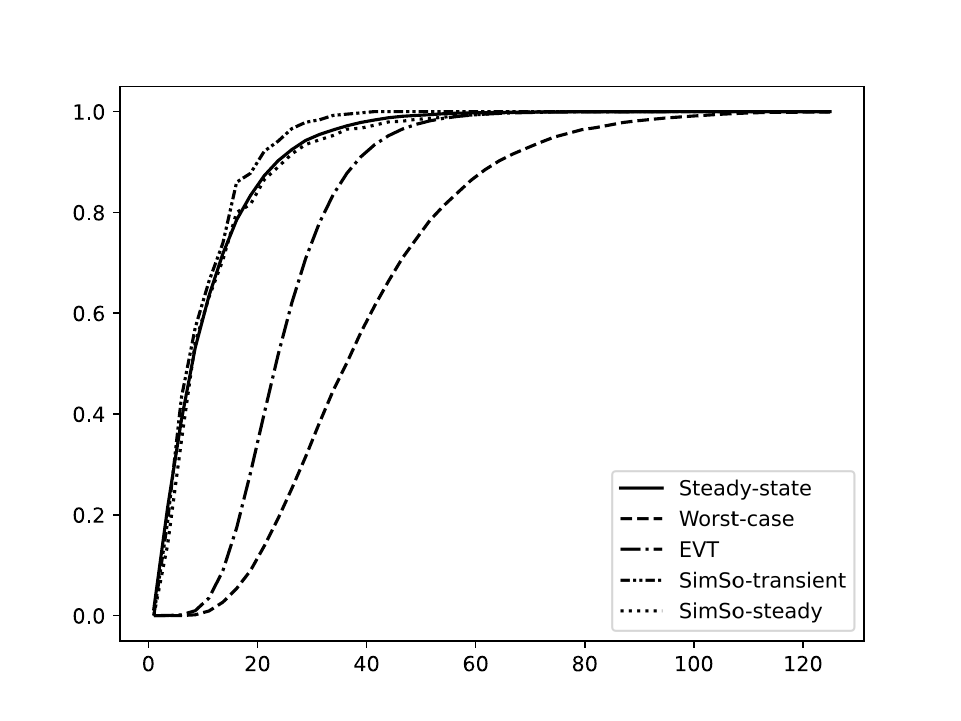}
		}
		\caption{Simulations of a $10\,000$ instances for the SimSo simulations, steady and transient, EVT estimation of the worst-case response time of the SimSo simulations and simulations of heavy-traffic worst-case response time and steady-state response-time when $\Ubar = 0.838$.}
		\label{fig:simulations}
	\end{figure*}

		In Lemma~\ref{eq:tda2} we prove that response times $R_{k,l}$ can be simulated from a sample of idle times of level $k-1$, a sample of execution times and the proper initialization backlog sample, the steady-case or the worst-case.
	\subsection{Generating heavy-traffic response times} \label{sec:sim}
	
 The procedure is as follows:
	
	\begin{itemize}
		\item \textbf{Generate} a backlog $b$ with the distribution function $\mu_{k-1}$ (resp. $\pi_{k-1}$),
		\item \textbf{Generate} an execution time $c$ with the  distribution function $F_k$,
		\item \textbf{Generate} the response time $R_{k} = \I^{b+c}_{k-1}$, with the probability function $\psi_{k-1}^{b+c}$.
	\end{itemize}
	
	See \figurename~\ref{fig:repidle} and \figurename~\ref{fig:sim_response_time}.

	\begin{proposition} \label{prop:RTdistribution}
		The probability function of $R_k$ is $h^{max}_{k}$ (resp. $\tilde{h}_k$).
	\end{proposition} 
	
	\begin{proof}
		It is a consequence of Lemma~\ref{coro:inverse}. Alternatively, we can see that by construction \eqref{eq:steadystate2} and \eqref{eq:steadystate} come from conditional probabilities. The representation in Lemma~\ref{eq:tda2} indicates that the distribution of response times  is conditioned by the values of the backlog and the execution time. 
	\end{proof}

	\subsection{Experimental results}

	The purpose of this section is to illustrate that the response times generated from Proposition~\ref{prop:RTdistribution} provide a good approximation, by comparing distribution functions of simulations and its associated EVT estimation, and generated heavy-traffic response times distributions. Closer are the curves, better is the estimation. Those results are not exhaustive and are an illustration, we do not cover in this work the sensitivity of the model for different values of $\Ubar$.  In order to do so, and to compare the distribution of the generated response times with the ones provided by SimSo, we modified the SimSo tool \cite{cheramy2014} to simulate probabilistic  systems and compare the scheduled (SimSo) simulations and heavy-traffic response times simulations. We use the data generated by SimSo to apply EVT on response times (using the Scipy framework\footnote{\url{https://docs.scipy.org/doc/scipy/reference/generated/scipy.stats.gereme.html}} on Python). Finally we compare our results with SimSo simulations and EVT estimations. For simplicity the periods are deterministic in these simulations.
	
	In order to illustrate the stability described in this paper, we use the task set $\Gamma$ provided in Table~\ref{tab:taskset}, with $\mathcal{C}_k$ being the set of possible values for the execution time $C_k$ and $p_k=dF_k/dx$ its probability function. The level $2$ maximum utilization is smaller than $1$, hence backlogs of level $2$ converge quickly to $0$, see \figurename~\ref{fig:stationary1}. The level $3$ maximum utilization is greater than $1$ and the level $3$ mean utilization is smaller than $1$. Thus, $\tau_3$ is the task of interest, see \figurename~\ref{fig:backlog2} and \figurename~\ref{fig:stationary2}. The level $4$ mean utilization is close to $1$, hence \figurename~\ref{fig:backlog3} and \figurename~\ref{fig:stationary3} illustrate the behavior of backlogs when the utilization approaches its phase transition. The level $5$ mean utilization is greater than $1$, which illustrates the \textit{explosion} (infinite response times) of the system, see \figurename~\ref{fig:backlog4} and \figurename~\ref{fig:stationary4}.
	
	In \figurename~\ref{fig:demand} we observe that the demand $W(t)$ follows the line $\Ubar t$ (its mean). In \figurename~\ref{fig:demand} and \ref{fig:backlog} we see what happens when the system mean utilization is smaller, close and greater than $1$ : for smaller values of $\Ubar$ the system stays with zero backlog at some point, see \figurename~\ref{fig:backlog2}, but for values greater than $1$ the system explodes, see \figurename~\ref{fig:backlog4} and \ref{fig:stationary4}. In \figurename~\ref{fig:backlog3} and \ref{fig:stationary3} we see that even for $\Ubar$ close to $1$, the system always admits finite idle times. Most importantly, we see in \figurename~\ref{fig:backlog2} and \ref{fig:stationary2} that when $\Ubar < 1$ and $\Umax > 1$, the analysis holds and provides quantifiable response times.
	
	In \figurename~\ref{fig:relaxation}, we see to what corresponds idle times and $\epsilon$-idle times graphically. In \figurename~\ref{fig:mean_response_time} the average idle time corresponds to the point where the line $t$ and $\Ubar t$ meet, and the distribution of the idle times correspond to the frequency the demand process meets the line $t$ for each instant $t > 0$.
	
	Finally, in \figurename~\ref{fig:sim_response_time}, we can see the simulations presented in Section~\ref{sec:sim} and a comparison with response times simulated via SimSo \cite{cheramy2014}. The two \textit{ground truth} samples are the two subsets \textit{SimSo-transient} and \textit{SimSo-steady}, which are composed respectively of simulated response time released before and after the maximum $\epsilon$-idle time $t_{idle}^{max}(\epsilon)=154$ where $\epsilon=10^{-6}$. The maximum idle time $t_{idle}^{max}(\epsilon)$ is computed via Monte-Carlo approximations with the representation in \eqref{eq:trelmax}. 

	The EVT estimation (green curve) from these SimSo simulations and the steady-state and worst-case response times suggested in this paper are compared via their distribution functions in \figurename~\ref{fig:rt_cdf}. In this case, where $\Ubar$ is not too close to $1$ and $\Umax$ greater than $1$, we can see that the \textit{Worst-case} response times, the \textit{EVT} estimation and the steady-state response times are greater (in the stochastic sense \cite{diaz2004pessimism}) than the \textit{true} response times simulated with SimSo. The heavy-traffic \textit{Worst-case} response time seems to be an upper-bound in practice, and the \textit{Steady-state} response time is quite accurate.

	Closer is the \textit{Steady-state} to the \textit{SimSo-steady} curve, more accurate the approximation is. We can see a big difference between the \textit{Worst-case} curve and the \textit{Steady-state} curve. However the proposed analysis does not permit to quantify analytically this difference.

	\begin{table}[t]
		\centering
		\caption{Task set used in the simulations of the experimental results.}
		\begin{tabular}{|c||c|c|c|c|c|}
			\hline
			task $\tau_k$ &  $T_k$ & $\mathcal{C}_k$ & $p_k$  & $\Ubar_{k}$ & $\Umax_{k}$\\
			\hline
			\hline
			$\tau_1$ & $4$ & $(1, 2)$ & $(0.5, 0.5)$  & $0.375$ & $0.5$\\
			\hline
			$\tau_2$ & $6$ & $(1, 2)$ & $(0.5, 0.5)$ &  $0.625$ & $0.833$\\
			\hline
			$\tau_3$ & $8$ & $(1, 2, 3)$ & $(0.5, 0.3, 0.2)$  &   $0.838$ & $1.208$\\
			\hline
			$\tau_4$ & $10$ & $(1, 2, 3)$ & $(0.6, 0.2, 0.2)$ & $0.998$ & $1.508$\\
			\hline
			$\tau_5$ & $12$ & $(1, 2, 3, 4)$ & $(0.5, 0.3, 0.1, 0.1)$ &  $1.148$ & $1.841$\\
			\hline
		\end{tabular}
		\label{tab:taskset}
	\end{table}

	\section{Future work}\label{sec:future}
	
	The results of this paper can be generalized if some assumptions of the model are relaxed. Indeed, we have seen that the steady-state only depends on inter-arrivals only through their mean value in the heavy-traffic analysis. This means that one could suppose inter-arrival independent, also known as \textit{Poisson arrivals}, and use the wide list of results on Poisson arrivals to extend the results of this paper.
	
	\subsection{Extension to dynamic-priority policies}
	
	The first step into dynamic scheduling, as for example Earliest Deadline First (EDF) is to study the processes $\hat{\beta}$ and $\hat{W}$ introduced in Section~\ref{sec:relaxation}. Indeed, for fixed-priority policies, those variables are simply the backlog and demand of the lowest priority level. However, for EDF, levels of priority need to be defined not only for tasks but for jobs. In \cite[Section 4.4]{diaz2002stochastic}, authors use the concept of \textit{ground jobs} which are jobs released at an instant where the system is idle, and as shown in the Loynes theorem, when $\Ubar < 1$, idle times are finite. This means that an analysis mixing the concept of ground jobs and idle times as defined in this paper can provide an extension of our results for dynamic-priority scheduling.
	
	\subsection{Extension to dependent and time-driven execution times} \label{sec:dependence}
	
	Considering time-driven execution times, \ie execution times with distribution functions changing over time, is possible. We used drifted Brownian motions to describe the demand. However, we used a model with constant coefficients. A more general class of drifted Brownian motion involve drift and deviation function, depending not only on the values of the demand but also on the the time interval they are being considered. This is a natural extension of the model presented in our work and leads to a more general type of real-time system that needs to be defined.
	
	Considering dependence between tasks is also possible. It would lead the analysis to multivariate Brownian motion which requires some new concepts but uses the same arguments we used in this work. The demand should describe not only priority levels but task levels, and cannot be superposed as we have done. Indeed, as a some properties are quite instinctive when independence is assumed (Wald's lemma and the law of large numbers for example), they need some workaround when this assumption is relaxed.
	
	\subsection{Extension to multi-core processors}
	
	The Loynes theorem also exists for the multi-core case \cite{baccelli2013elements}, which means that there is a way to apply the results we provide in this paper to a multi-core model, for a global fixed-priority scheduling policy, and a global dynamic-scheduling policy. Furthermore, for independent inter-arrivals, optimal algorithm already exist for priority assignment and processor assignment, \cf \cite{huang2015control, boxma1989workloads, van1995dynamic, baccelli2013elements, chen2001fundamentals}.
	
	\section{Conclusion}
	
	We have seen that real-time systems can reach steadiness over a finite and quantifiable amount of time, and that a necessary condition to assure this stability is that the mean utilization is smaller than $1$. In practice, systems with a mean utilization $\Ubar > 0.9$ have many deadline misses, which requires from system designers to quantify deadline miss probabilities carefully by using the distributions of response times provided in this paper. 
	No schedulability tests considering the steadiness of response times exist. This is a natural step in our opinion for the use of the analysis provided in this work. For example, the analysis we have presented in this paper is well suited for an application of a Monte-Carlo response time analysis \cite{bozhko2021monte} which has recently been proven efficient.
	
	We have expressed the probability function of response times in a specific family of distributions, the inverse Gaussian distribution, and provided a method to generate them. However, the distribution functions of execution time need to be known, but they are usually unknown. The methods built in this paper could be used in empirical and measurement-based methods, for example using clustering methods \cite{friebe2020identification,zagalo2020identification}. 

	\bibliography{biblio.bib} 
	\vspace{-1.2cm}

	\begin{IEEEbiographynophoto}{Kevin Zagalo} is a PhD. student under the supervision of Liliana Cucu-Grosjean and Avner Bar-hen in the Kopernic team at Inria Paris research institute, working on probabilistic real-time systems. After studying computer science at Université de Paris, he graduates a Master degree in \textit{Probabilities and Stochastic analysis} at Sorbonne Université.
	\end{IEEEbiographynophoto}
	\vspace{-1.2cm}
	\begin{IEEEbiographynophoto}{Yasmina Abdeddaïm} is an associate professor at Gustave Eiffel University (UGE) and a member of the Laboratoire d'Informatique Gaspard Monge research laboratory (LIGM). She is also an associate member of the Kopernic team at the Inria Paris research institute.% where she collaborates with the team on the development of probabilistic methods for the validation of real-time systems.
	\end{IEEEbiographynophoto}
	\vspace{-1.2cm}
	\begin{IEEEbiographynophoto}{Avner Bar-Hen} is a professor at the Conservatoire des Arts et Métiers (Cnam) in Paris, holding the chair of \textit{Statistics and Big Data}. He is currently working in the Kopernic team at the Inria Paris research institute on stochastic approaches for real-time systems.
	\end{IEEEbiographynophoto}
	\vspace{-1.2cm}
	\begin{IEEEbiographynophoto}{Liliana Cucu-Grosjean} is a research director at Inria and the head of the Kopernic team. She actively supports the academic and industrial use of probabilistic approaches for the design of mixed-criticality real-time systems.
    \end{IEEEbiographynophoto}%

\end{document}